\documentclass[11pt]{article}
\usepackage[margin=1.0in]{geometry}

\usepackage{cite}
\usepackage{float}
\usepackage{amsmath,amssymb,amsfonts}
\usepackage{algorithm,algorithmic}
\usepackage{amssymb}
\usepackage{amsmath}
\usepackage{amsthm}
\usepackage{graphicx}
\usepackage[table]{xcolor}
\usepackage{hyperref}
\usepackage{pdfpages}
\usepackage{booktabs}
\usepackage{caption}
\usepackage{subcaption}
\usepackage{authblk}
\usepackage{tikz}
\usetikzlibrary{arrows,shapes,shapes.geometric,arrows,fit,calc,positioning,automata,chains,positioning,quotes}
\usepackage{multirow}
\def\BibTeX{{\rm B\kern-.05em{\sc i\kern-.025em b}\kern-.08em
    T\kern-.1667em\lower.7ex\hbox{E}\kern-.125emX}}

\newtheorem{thm}{Theorem}
\newtheorem{definition}{Definition}
\newtheorem{problem}{Problem}
\newtheorem{remark}{Remark}
\newtheorem{lemma}{Lemma}

\newcommand{\state}{\ensuremath{\mathcal{S}}}
\newcommand{\action}{\ensuremath{\mathcal{A}}}
\newcommand{\transition}[1]{\ensuremath{T_{#1}}}
\newcommand{\probmat}{\ensuremath{\mathbf{P}}}
\newcommand{\lbprobmat}{\ensuremath
{\check{\probmat}}}
\newcommand{\ubprobmat}{\ensuremath{\hat{\probmat}}}
\newcommand{\pertmat}{\ensuremath{\mathbf{X}}}

\newcommand{\pertset}[1]{\ensuremath{\mathcal{PS}^{#1}}}
\newcommand{\mdp}{\ensuremath{\mathcal{M}}}
\newcommand{\mc}{\ensuremath{\mathcal{C}}}
\newcommand{\prsat}[1]{\ensuremath{Pr(s_0\models_{#1} \varphi)}}

\newcommand{\rational}{\ensuremath{\mathbb{Q}}}

\newcommand{\paramset}{\ensuremath{X}}
\newcommand{\etal}{\emph{et. al}}

\title{Adversarial Robustness Verification and Attack Synthesis in Stochastic Systems\thanks{To appear in \textit{35th IEEE Computer Security Foundations Symposium (2022)}}}

\author{Lisa Oakley \qquad Alina Oprea \qquad Stavros Tripakis}
\affil{ \small \emph{Khoury College of Computer Sciences, Northeastern University }}

\date{\vspace{-5ex}}

\begin{document}
\maketitle

\begin{abstract}
    Probabilistic model checking is a useful technique for specifying and verifying properties of stochastic systems including randomized protocols and reinforcement learning models. However, these methods rely on the assumed structure and probabilities of certain system transitions. These assumptions may be incorrect, and may even be violated by an adversary who gains control of some system components.

    In this paper, we develop a formal framework for adversarial robustness in systems modeled as discrete time Markov chains (DTMCs). We base our framework on existing methods for verifying probabilistic temporal logic properties and extend it to include deterministic, memoryless policies acting in Markov decision processes (MDPs). Our framework includes a flexible approach for specifying structure-preserving and non structure-preserving adversarial models several adversarial models with different capabilities to manipulate the system. We outline a class of threat models under which adversaries can perturb system transitions, constrained by an $\varepsilon$ ball around the original transition probabilities.

    We define three main DTMC adversarial robustness problems: adversarial robustness verification, maximal $\delta$ synthesis, and worst case attack synthesis. We present two optimization-based solutions to these three problems, leveraging traditional and parametric probabilistic model checking techniques. We then evaluate our solutions on two stochastic protocols and a collection of Grid World case studies, which model an agent acting in an environment described as an MDP. We find that the parametric solution results in fast computation for small parameter spaces. In the case of less restrictive (stronger) adversaries, the number of parameters increases, and directly computing property satisfaction probabilities is more scalable. We demonstrate the usefulness of our definitions and solutions by comparing system outcomes over various properties, threat models, and case studies.
\end{abstract}
\section{Introduction}\label{sec:intro}
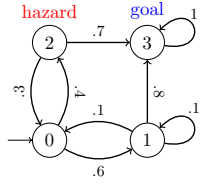
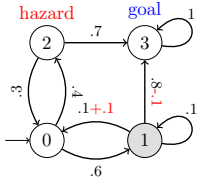
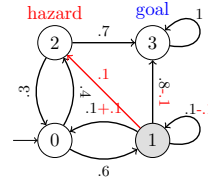
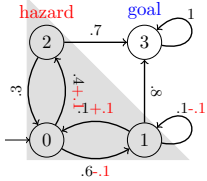
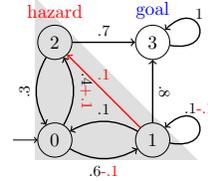
\begin{figure}[ht]
    \centering

\begin{subfigure}[b]{\linewidth}
    \centering
    \scalebox{.65}{%
    \begin{tikzpicture}[main/.style = {draw, circle},node distance = 2cm, initial text =,state/.style={circle, draw, minimum size=.7cm}]] 
        \node[state,initial] (0) {$0$}; 
        \node[state] (1) [right of=0] {$1$};
        \node[state] (2) [above of=0,label=above:{\normalsize \textcolor{red}{hazard}}] {$2$};
        \node[state] (3) [above of=1,label={\normalsize \textcolor{blue}{goal}}] {$3$};

        % 0 up
        \draw[->, thick] (0) [out=60,in=-60] to node [above,midway,rotate=-90] {\footnotesize .4} node [swap] {} (2);
        % 0, right
        \draw[->, thick] (0) [out=-30,in=180+30] to node [below,midway] {\footnotesize .6} node [swap] {} (1);

        % 1 up
        \draw[->, thick] (1) [out=90,in=-90] to node [above,midway,rotate=-90] {\footnotesize .8} node [swap] {} (3);
        % 1 left
        \draw[->, thick] (1) [out=180-30,in=30] to node [above,midway,rotate=0] {\footnotesize .1} node [swap] {} (0);
        % 1 self
        \draw[->, thick] (1) [out=-10,in=45,loop] to node [above,midway, yshift=3pt] {\footnotesize .1} node [swap] {} (1);

        % 2 down
        \draw[->, thick] (2) [out=270-30,in=90+30] to node [above,midway,rotate=90] {\footnotesize .3} node [swap] {} (0);
        % 2 right
        \draw[->, thick] (2) [out=0,in=180] to node [above,midway,rotate=0] {\footnotesize .7} node [swap] {} (3);

        % 3 self
        \draw[->, thick] (3) [out=-10,in=45,loop] to node [above,midway,yshift=1pt] {\footnotesize 1} node [swap] {} (3);
    \end{tikzpicture} 
    }
    \caption{Original DTMC $\mc=(\state,s_0,\probmat)$.}
\end{subfigure}

\begin{subfigure}[b]{.42\linewidth}
    \centering
    \scalebox{.65}{%
    \begin{tikzpicture}[main/.style = {draw, circle},node distance = 2cm, initial text =,state/.style={circle, draw, minimum size=.7cm}]] 
        \node[state,initial] (0) {$0$}; 
        \node[state] (1) [right of=0,fill=lightgray!50] {$1$};
        \node[state] (2) [above of=0,label=above:{\normalsize \textcolor{red}{hazard}}] {$2$};
        \node[state] (3) [above of=1,label={\normalsize \textcolor{blue}{goal}}] {$3$};

        % 0 up
        \draw[->, thick] (0) [out=60,in=-60] to node [above,midway,rotate=-90] {\footnotesize .4} node [swap] {} (2);
        % 0, right
        \draw[->, thick] (0) [out=-30,in=180+30] to node [below,midway] {\footnotesize .6} node [swap] {} (1);

        % 1 up
        \draw[->, thick] (1) [out=90,in=-90] to node [above,midway,rotate=-90] {\footnotesize .8\textcolor{red}{-.1}} node [swap] {} (3);
        % 1 left
        \draw[->, thick] (1) [out=180-30,in=30] to node [above,midway,rotate=0] {\footnotesize .1\textcolor{red}{+.1}} node [swap] {} (0);
        % 1 self
        \draw[->, thick] (1) [out=-10,in=45,loop] to node [above,midway, yshift=3pt] {\footnotesize .1} node [swap] {} (1);

        % 2 down
        \draw[->, thick] (2) [out=270-30,in=90+30] to node [above,midway,rotate=90] {\footnotesize .3} node [swap] {} (0);
        % 2 right
        \draw[->, thick] (2) [out=0,in=180] to node [above,midway,rotate=0] {\footnotesize .7} node [swap] {} (3);

        % 3 self
        \draw[->, thick] (3) [out=-10,in=45,loop] to node [above,midway,yshift=1pt] {\footnotesize 1} node [swap] {} (3);
    \end{tikzpicture} 
    }
    \caption{Worst case SPSS attack with vulnerable state 1.}
\end{subfigure}
\qquad
\begin{subfigure}[b]{.42\linewidth}
    \centering
    \scalebox{.65}{%
    \begin{tikzpicture}[main/.style = {draw, circle},node distance = 2cm, initial text =,state/.style={circle, draw, minimum size=.7cm}]] 
        \node[state,initial] (0) {$0$}; 
        \node[state] (1) [right of=0,fill=lightgray!50] {$1$};
        \node[state] (2) [above of=0,label=above:{\normalsize \textcolor{red}{hazard}}] {$2$};
        \node[state] (3) [above of=1,label={\normalsize \textcolor{blue}{goal}}] {$3$};

        % 0 up
        \draw[->, thick] (0) [out=60,in=-60] to node [above,midway,rotate=-90] {\footnotesize .4} node [swap] {} (2);
        % 0, right
        \draw[->, thick] (0) [out=-30,in=180+30] to node [below,midway] {\footnotesize .6} node [swap] {} (1);

        % 1 up
        \draw[->, thick] (1) [out=90,in=-90] to node [above,midway,rotate=-90] {\footnotesize .8\textcolor{red}{-.1}} node [swap] {} (3);
        % 1 left
        \draw[->, thick] (1) [out=180-30,in=30] to node [above,midway,rotate=0] {\footnotesize .1\textcolor{red}{+.1}} node [swap] {} (0);
        % 1 self
        \draw[->, thick] (1) [out=-10,in=45,loop] to node [above,midway, yshift=3pt] {\footnotesize .1\textcolor{red}{-.1}} node [swap] {} (1);
        % 1 diag
        \draw[draw=red,->, thick] (1) to node [above,midway, yshift=3pt] {\footnotesize \textcolor{red}{.1}} node [swap] {} (2);

        % 2 down
        \draw[->, thick] (2) [out=270-30,in=90+30] to node [above,midway,rotate=90] {\footnotesize .3} node [swap] {} (0);
        % 2 right
        \draw[->, thick] (2) [out=0,in=180] to node [above,midway,rotate=0] {\footnotesize .7} node [swap] {} (3);

        % 3 self
        \draw[->, thick] (3) [out=-10,in=45,loop] to node [above,midway,yshift=1pt] {\footnotesize 1} node [swap] {} (3);
    \end{tikzpicture} 
    }
    \caption{Worst case SS attack with vulnerable state 1.}
\end{subfigure}

\begin{subfigure}[b]{.42\linewidth}
    \centering
    \scalebox{.65}{%
    \begin{tikzpicture}[main/.style = {draw, circle},node distance = 2cm, initial text =,state/.style={circle, draw, minimum size=.7cm}]] 
        \draw (-.4,-.4) [fill=lightgray!50,draw=lightgray!50] node {}
        -- (-.4,2.9) node {}
        -- (2.9,-.4) node {}
        -- cycle;
        \node[state,initial] (0) [] {$0$}; 
        \node[state] (1) [right of=0] {$1$};
        \node[state] (2) [above of=0,label=above:{\normalsize \textcolor{red}{hazard}}] {$2$};
        \node[state] (3) [above of=1,label={\normalsize \textcolor{blue}{goal}}] {$3$};

        % 0 up
        \draw[->, thick] (0) [out=60,in=-60] to node [above,midway,rotate=-90] {\footnotesize .4\textcolor{red}{+.1}} node [swap] {} (2);
        % 0, right
        \draw[->, thick] (0) [out=-30,in=180+30] to node [below,midway] {\footnotesize .6\textcolor{red}{-.1}} node [swap] {} (1);

        % 1 up
        \draw[->, thick] (1) [out=90,in=-90] to node [above,midway,rotate=-90] {\footnotesize .8} node [swap] {} (3);
        % 1 left
        \draw[->, thick] (1) [out=180-30,in=30] to node [above,midway,rotate=0] {\footnotesize .1\textcolor{red}{+.1}} node [swap] {} (0);
        % 1 self
        \draw[->, thick] (1) [out=-10,in=45,loop] to node [above,midway, yshift=3pt] {\footnotesize .1\textcolor{red}{-.1}} node [swap] {} (1);

        % 2 down
        \draw[->, thick] (2) [out=270-30,in=90+30] to node [above,midway,rotate=90] {\footnotesize .3} node [swap] {} (0);
        % 2 right
        \draw[->, thick] (2) [out=0,in=180] to node [above,midway,rotate=0] {\footnotesize .7} node [swap] {} (3);

        % 3 self
        \draw[->, thick] (3) [out=-10,in=45,loop] to node [above,midway,yshift=1pt] {\footnotesize 1} node [swap] {} (3);
    \end{tikzpicture} 
    }
    \caption{Worst case SPST attack with transitions between states $0,\;1,\;2$ vulnerable.}
\end{subfigure}
\qquad
\begin{subfigure}[b]{.42\linewidth}
    \centering
    \scalebox{.65}{%
    \begin{tikzpicture}[main/.style = {draw, circle},node distance = 2cm, initial text =,state/.style={circle, draw, minimum size=.7cm}]] 
        \draw (-.4,-.4) [fill=lightgray!50,draw=lightgray!50] node {}
        -- (-.4,2.9) node {}
        -- (2.9,-.4) node {}
        -- cycle;
        \node[state,initial] (0) [] {$0$}; 
        \node[state] (1) [right of=0] {$1$};
        \node[state] (2) [above of=0,label=above:{\normalsize \textcolor{red}{hazard}}] {$2$};
        \node[state] (3) [above of=1,label={\normalsize \textcolor{blue}{goal}}] {$3$};

        % 0 up
        \draw[->, thick] (0) [out=60,in=-60] to node [above,midway,rotate=-90] {\footnotesize .4\textcolor{red}{+.1}} node [swap] {} (2);
        % 0, right
        \draw[->, thick] (0) [out=-30,in=180+30] to node [below,midway] {\footnotesize .6\textcolor{red}{-.1}} node [swap] {} (1);

        % 1 up
        \draw[->, thick] (1) [out=90,in=-90] to node [above,midway,rotate=-90] {\footnotesize .8} node [swap] {} (3);
        % 1 left
        \draw[->, thick] (1) [out=180-30,in=30] to node [above,midway,rotate=0] {\footnotesize .1} node [swap] {} (0);
        % 1 self
        \draw[->, thick] (1) [out=-10,in=45,loop] to node [above,midway, yshift=3pt] {\footnotesize .1\textcolor{red}{-.1}} node [swap] {} (1);
        % 1 diag
        \draw[draw=red,->, thick] (1) to node [above,midway, yshift=3pt] {\footnotesize \textcolor{red}{.1}} node [swap] {} (2);

        % 2 down
        \draw[->, thick] (2) [out=270-30,in=90+30] to node [above,midway,rotate=90] {\footnotesize .3} node [swap] {} (0);
        % 2 right
        \draw[->, thick] (2) [out=0,in=180] to node [above,midway,rotate=0] {\footnotesize .7} node [swap] {} (3);

        % 3 self
        \draw[->, thick] (3) [out=-10,in=45,loop] to node [above,midway,yshift=1pt] {\footnotesize 1} node [swap] {} (3);
    \end{tikzpicture} 
    }
    \caption{Worst case ST attack with transitions between states $0,\;1,\;2$ vulnerable.}
\end{subfigure}
    \caption{Worst case attack on DTMC $\mc$ with respect to $\varphi=((s\not=2) \mathbf{U}^{\leq 10} (s=3))$ over four $\varepsilon,max$-bounded threat models, with perturbation budget $\varepsilon=.1$. Selected states (SS) and structure-preserving selected states (SPSS) adversaries may perturb transition probabilities which come from a set of ``vulnerable states.'' Selected transition (ST) and structure-preserving selected transition (SPST) adversaries may only perturb transition probabilities in some pre-defined set of transitions. SS and ST adversaries can add transitions within their vulnerability sets, and SPSS and SPST adversaries can only perturb transitions which existed in the original DTMC.}
    \label{fig:attack}
\end{figure}

\begin{figure*}[ht]
    \centering
    \includegraphics[width=.6\linewidth]{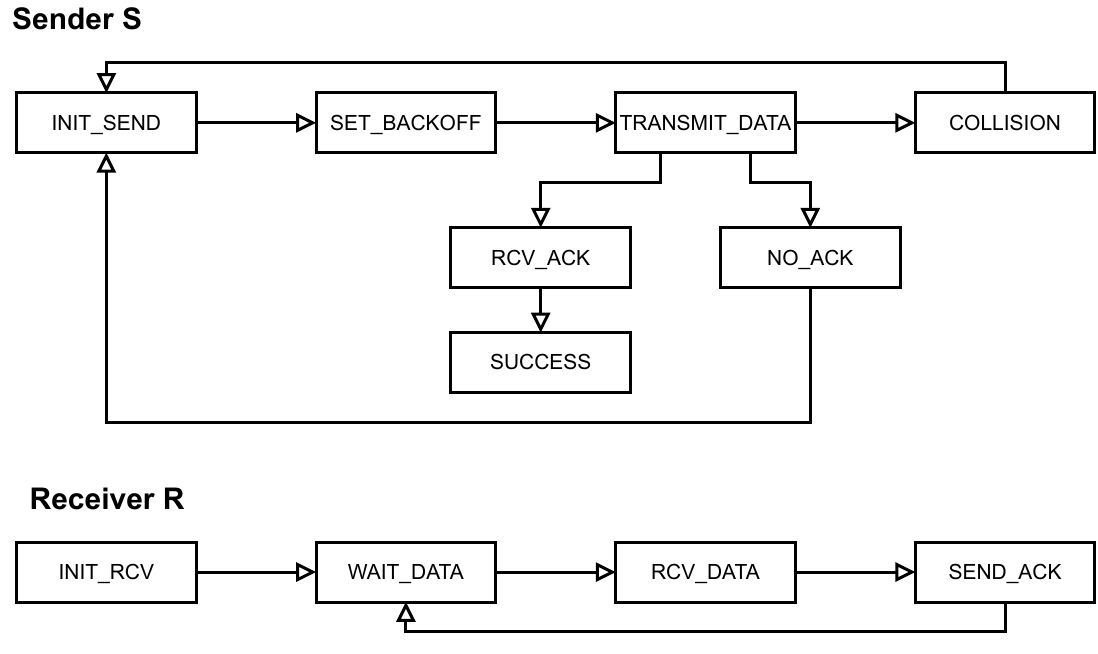}
    \caption{Probabilistic collision avoidance protocol model with exponential backoff.}
    \label{fig:backoff}
\end{figure*}

Specification and verification of stochastic systems including discrete time Markov chains (DTMCs) and Markov decisions processes (MDPs) can be used to analyze a large class of probabilistic systems and algorithms~\cite{KNP12a,DKNP06,KN02,daws_symbolic_2005}. A promising recent line of research has also used probabilistic model checking to verify properties of deep reinforcement learning~\cite{gu_demonstration_2020,gotsman_deep_2020}. The majority of work on property verification on these systems considers fixed, known probabilities of transitions between states. In real-world scenarios, these assumed transition probabilities may be incorrect or manipulated by an adversary who gains control of some or all components of the system.

In this paper, we use existing frameworks for verification of stochastic systems to develop a formal definition for adversarial robustness in systems defined as DTMCs with respect to a wide variety of properties~\cite{storm,param,prophesy}. 
Under our definition, an adversary can generate bounded perturbations, structured as a matrix, on a DTMC or MDP with the goal of decreasing the probability of satisfying certain logical properties. We model the constraints on the perturbation matrix as adversarial perturbation sets, creating a flexible framework to define threat models. We specify a class of \emph{$\varepsilon,d$-bounded} threat models for which the perturbation matrix is bounded by some scalar $\varepsilon$ with respect to some metric distance $d$.

In the past, work has been done to address the issue of uncertain transition probabilities in DTMCs. The majority of these works focus on uncertainty in models with small parameter spaces, considering only a limited set of properties ~\cite{winkler_complexity_2019,quatmann_parameter_2016,prophesy}. In the past few years, improvements have been made to these techniques making probabilistic model checking in the presence of uncertainty more feasible for large-scale systems with many uncertainties ~\cite{junges2019parameter,cubuktepe_synthesis_2018}. However, most existing work considers only non-adversarial uncertainty, and does not account for situations in which an adversary may gain control of certain system components in order to intentionally reduce utility of the system. Moreover, most analysis and evaluation has assumed uncertainties which preserve the structure of the original DTMC. We refer to models of attackers who exploit this kind of uncertainty as \emph{structure-preserving} threat models. Since the model probabilities are often generated by observation of the system, they do not account for unexplored malicious transitions. In our paper, we additionally consider this scenario, where an attacker is able to add transitions to the system (i.e., increase the probability of these transitions up from zero). We refer to these attackers as \emph{non structure-preserving}. Fig. \ref{fig:attack} shows examples of   attacks under structure-preserving and non structure-preserving $\varepsilon,max$-bounded threat models on a 4-state DTMC.

Providing a flexible formalism which can model a wide range of attackers has real world importance. For example, consider a protocol which implements collision avoidance using exponential backoff \cite{KNSW07,KNS02a,Fru11}. In Fig. \ref{fig:backoff}, we model the components of a protocol with multiple senders and receivers who communicate over a wireless channel such that senders cannot send messages simultaneously. In the event of a a collision (two senders send a message at the same time), neither message is sent, and both senders must make another attempt to send the message. In this model, senders use a probabilistic backoff mechanism to reduce the probability of collision. After receiving the sender's message, the receiver sends an acknowledgment and the sender moves to RCV\_ACK.

We can consider the probability that a message is successfully sent and received within $n$ time ticks as a reasonable measure of utility for the collision avoidance protocol. Under this model, a structure-preserving attacker with access to a sender could increase the backoff by decreasing the probability that a message is sent, thereby unnecessarily slowing the time between message sends, and decreasing the overall utility. However, considering only structure-preserving attackers is limiting. For example, an attacker could gain access to a receiver and trigger arbitrary receipt acknowledgments, increasing the probability that the sender moves to the RCV\_ACK state directly from the SET\_BACKOFF or COLLISION state and bypassing the TRANSMIT\_DATA state entirely. This would effectively result in a dropped message, causing a significant decrease in system utility.

In addition to probabilistic protocols, we consider the cyber-physical scenario of a robot moving around adjacent spaces in a grid to bring packages from one side of a warehouse to the other. This is a variation of the Grid World problem commonly used to evaluate reinforcement learning algorithms \cite{gridworld}. The grid can contain fixed ``hazard'' states which the robot incurs damage by entering (such as an open stairwell) and obstacles which the robot cannot pass through (such as a shelf or barrel). At each time tick, the robot chooses a preferred direction to move based on a fixed policy. However, based on probabilistic environmental factors such as another robot blocking an adjacent space, the actual direction the robot moves after choosing an action is modeled probabilistically based on observations of the environment over time. In this scenario, we can model utility as the probability that a robot picks up a package at the warehouse entrance (bottom left corner) and reaches the goal shelf (top left corner) within a fixed time interval while avoiding all hazard states. Here, a structure-preserving attacker can cause the robot to take a longer path on average by increasing the probability that another robot is in its way. A non structure-preserving attacker could additionally remove a shelf or barrel to introduce new hazard states which were previously unreachable. In both the collision avoidance and Grid World examples, our threat model definition allows for modeling of both structure-preserving and non structure-preserving scenarios, providing a more flexible and powerful attacker formalism than in prior work.

After specifying our classes of threat models, we go on to define three main adversarial robustness problems. In the \emph{DTMC adversarial robustness verification} problem, we must determine whether a given DTMC satisfies adversarial robustness with respect to a given property, bound, and adversarial model. In the \emph{maximal $\delta$ synthesis} problem, we are given a DTMC, property, and adversarial model, and must find the maximal bound $\delta$ for which the DTMC is adversarially robust. In the \emph{worst-case attack synthesis} problem, we are given a DTMC, property, and adversarial model, and must find a perturbation to the original DTMC which results in the largest $\delta$ between the probability that the original and perturbed DTMCs satisfy the property. We refer to these problems collectively as the \emph{DTMC adversarial robustness problems}. We also extend the robustness definitions to agents acting in MDPs under deterministic, memoryless policies. These problems are important because they allow a system designer to audit their systems from an adversarial perspective, highlighting vulnerabilities and providing feedback on which components are most impactful when under attack. 

To position our problem definitions in the context of other literature, we compare to the uncertain Markov chain (UMC) model checking problem \cite{hermanns_model-checking_2006}. We show that DTMC adversarial robustness \emph{verification} can be reduced to the UMC model checking problem for interval-valued discrete time Markov chains (IDTMC) in the case of $\varepsilon,max$-bounded threat models and PCTL properties. Our definitions support more properties, including PCTL$^*$, and our worst case adversary synthesis and maximal $\delta$ synthesis problems provide the added benefit of explaining which adversarial perturbation is a counterexample to robustness.

After defining our problem space, we outline two optimization-based solutions which use the Sequential Least Squares (SLSQP) minimization algorithm~\cite{2020SciPy-NMeth} and existing probabilistic model checking techniques for explicit and parametric DTMCs. In particular, we use the Prism \cite{prism} probabilistic model checker for computing explicit property satisfaction probabilities, and the PARAM~\cite{param} parametric model checking tool to compute a closed form, symbolic representation of the property satisfaction probability.  We evaluate these solutions on two protocols modeled as DTMCs \cite{baier2008principles,daws_symbolic_2005} and one Grid World MDP \cite{gridworld}. Our experiment code is publicly available: \url{https://github.com/lisaoakley/dtmc_attack_synthesis}. We find that symbolic methods work well when the parameter space is small, but that direct computation of property satisfaction probabilities is faster for larger parameter spaces. In our evaluation, we show how our methods can be useful in determining which system components and transitions are most important to preserving robustness of the system across many properties and threat models. We also compare various threat models for different case studies and properties. In Section \ref{sec:adv_ml}, we draw a connection between our problem definition and solutions to the notion of adversarial robustness in supervised machine learning models with respect to deployment-time adversarial examples \cite{biggio2013evasion,szegedy2013intriguing}.

We have four main contributions in this paper. First, we define adversarial robustness for DTMCs with a flexible threat model definition including structure-preserving and non structure-preserving adversaries. Second, we propose three DTMC adversarial robustness problems and provide two solutions to these problems based on existing probabilistic model checking techniques. Third, we provide theoretical reduction to the UMC model checking problem \cite{hermanns_model-checking_2006} for our proposed verification problem with $\varepsilon,max$-bounded threat models. Finally, we evaluate these definitions and solutions on three case studies and compare adversarial models with different capabilities.
\section{Background}\label{sec:background}
We start with a brief background on probabilistic model checking, including the definitions of discrete time Markov chains (DTMCs), Markov decision processes (MDPs), parametric DTMCs (pDTMCs) and the PCTL$^*$ property specification language \cite{baier2008principles,daws_symbolic_2005}. These modeling tools will be useful in our definitions of robustness, as well as the evaluation of different model checking techniques. We also point out the relationship between a policy, an MDP, and the resulting composed DTMC.

\subsection{Probabilistic Models}\label{sec:prob_models}
DTMCs are transition systems where transitions between states are represented by probability distributions. MDPs are transition systems which can model both probabilistic transitions and nondeterminism. An MDP composed with a deterministic, memoryless policy results in a DTMC, as the policy resolves all nondeterminism in the system.
\begin{definition}[DTMC]
    A \emph{discrete time Markov chain (DTMC)} is a 5-tuple 
    \begin{equation}
        \nonumber\mc=(\state,s_0,\probmat,AP,L)
    \end{equation}
    where \state{} is a finite set of states with $s_0\in \state$ the initial state, $AP$ a set of atomic propositions, and $L:\state{}\rightarrow 2^{AP}$ a labeling function. \probmat{} is an $|\state|\times|\state|$ matrix where $0\leq\probmat_{s,s'}\leq 1$ indicates the probability that the system transitions from state $s$ to $s'$. $\forall s\in \state, \;\sum_{s'\in\state}\probmat_{s,s'}=1$ (all rows in \probmat{} sum to 1).
\end{definition}

\begin{definition}[MDP]
    A \emph{Markov decision process (MDP)} is a 6-tuple 
    \begin{equation}
        \nonumber\mdp{}=(\state,s_0,\action,\transition{},AP,L)
    \end{equation}
    where $\state{},\action{}$ are finite sets of states and actions respectively, $s_0\in \state$ is the initial state, $AP$ a set of atomic propositions, and $L:\state{}\rightarrow 2^{AP}$ a labeling function. $\transition{}:\state\times\action\times\state\rightarrow [0,1]$ is a transition function where $\transition{}(s,a,s')$ represents the probability that the system transitions from $s$ to $s'$, given action $a$.
\end{definition}
\begin{definition}[DTMC of an MDP]
    The \emph{DTMC of an MDP}, $\mdp=(\state,s_0,\action,\transition,AP,L)$, induced by a policy $\sigma$ is
    \begin{equation}
        \nonumber\mdp_\sigma=(\state,s_0,\probmat,AP,L)
    \end{equation}
    where $\probmat_{s,s'}=\transition{}(s,\sigma(s),s')\;\forall s,s'\in\state$, and $\sigma:\state{}\rightarrow\action{}$ is a memoryless, deterministic policy.\label{def:DTMCofMDP}
\end{definition}

In some cases, we will need to consider a DTMC for which some or all of the transitions are of unspecified probability. To do so, we consider a parametric DTMC which describes some transition probabilities in the system symbolically, rather than with an explicit value \cite{daws_symbolic_2005}. 
    
\begin{definition}[Parametric DTMC (pDTMC)]
A (labeled) \emph{parametric DTMC} of formal parameter set \paramset{} can be described as a 5-tuple
\begin{equation}
    \mc_\paramset=(\state, s_0, \probmat^\paramset, AP, L)
\end{equation}
Where \state\ is a set of states, $s_0$ an initial state, AP a set of atomic propositions, $L:\state\rightarrow2^{AP}$ a labeling function, and $\probmat^\paramset$ a probability matrix such that $\probmat^\paramset_{s,s'}$ indicates the probability that the system transitions from state $s$ to $s'$. 
\end{definition}

In a parametric DTMC, the entries of $\probmat^\paramset$ can be rational numbers, members of \paramset, or expressions formed from the former two categories. For a parametric DTMC to be valid, there must exist an instantiation function $\kappa:\rational\cup \paramset\rightarrow[0,1]$ such that $\forall q\in\rational\;\kappa(q)=q$ and $\forall s\in\state\;\sum_{s'\in\state}\kappa(\probmat^\paramset_{s,s'})=1$. The resulting DTMC of applying $\kappa$ to the variables is called an \emph{instance} DTMC. 

\subsection{Property Specification}\label{sec:PCTL*}
In our analysis, we will primarily consider properties specified in PCTL*, an extension of Probabilistic Computation Tree Logic (PCTL), which is a formal logic that defines probabilistic properties on DTMCs. The PCTL* syntax is as follows~\cite{baier2008principles}.
\begin{definition}[Syntax of PCTL*]
    PCTL* \emph{state formulae} over the set AP of atomic propositions are formed according to the following grammar:
    \begin{equation}
        \Phi ::= true \; \mid \; a \; \mid \; \Phi_1\land\Phi_2 \; \mid \; \neg\Phi \; \mid \; \mathbb{P}_J(\varphi)
    \end{equation}
    where $a\in AP$, $\varphi$ is a path formula, and $J\subseteq [0,1]$ is an interval with rational bounds. PCTL* \emph{path formulae} are formed according to the following grammar:
    \begin{equation}
        \varphi ::= \Phi \; \mid \; \varphi_1\land\varphi_2  \; \mid \; \neg\varphi \; \mid \;\bigcirc\varphi \; \mid \; \varphi_1 \mathbf{U}\varphi_2 \; \mid \; \varphi_1 \mathbf{U}^{\leq k}\varphi_2
    \end{equation}
    where $\Phi$ is a PCTL* state formula.
\end{definition}
Using this syntax, we can derive the eventually ($\lozenge\Phi = true \mathbf{U}\Phi$) and bounded eventually ($\lozenge^{\leq k}\Phi = true \mathbf{U}^{\leq k}\Phi$) operators, as well as many other useful operators, like the global or always true operator ($\square\Phi=\neg\lozenge\neg\Phi$). We note that PCTL* subsumes PCTL whose definition is similar, but requires that any temporal operator be followed by a state formula, and restricts boolean combinations of path formulae. We omit the details of these definitions and refer the reader to Chapter 10 of \cite{baier2008principles} for more details.

\subsection{Computing Property Satisfaction Probabilities}\label{sec:background_satisfaction}
PCTL* introduces a probabilistic operator $\mathbb{P}_J(\varphi)$ where $\varphi$ is a path formula and $J\subseteq [0,1]$ is a probability interval. Semantically, given a state $s$ of a Markov chain $\mc$, we have that 
\begin{equation}
    s\models_\mc\mathbb{P}_J(\varphi) \iff Pr(s\models_\mc\varphi)\in J
\end{equation}
where $Pr(s \models_\mc \varphi)$ denotes the probability that s satisfies $\varphi$ in $\mc$. We refer to this measure as the \emph{satisfaction probability of $\varphi$ at s}. The satisfaction probability is defined by considering the set $S$ of all infinite traces starting from $s$ that satisfy $\varphi$, and then taking the probability measure of $S$.

Because we are analyzing the adversarial robustness of a system, we are interested in the extent to which small changes in the environment affect the utility of the system. In our case, the utility of the system is the probability that a property is satisfied by the paths starting in the initial state, $s_0$. Therefore, rather than focusing on the binary question of whether the PCTL* property is satisfied by the DTMC, we are interested in the satisfaction probabilities of certain path formulae. In the remainder of the paper, we utilize the direct computation of $Pr(s\models_\mc\varphi)$. State of the art probabilistic model checkers including PRISM~\cite{prism}, ISCASMC~\cite{ISCASMC} and Storm~\cite{storm} are capable of computing this probability directly. 

Parametric model checkers including PARAM~\cite{param} and PROPhESY~\cite{prophesy} are able to determine these probabilities for parametric DTMCs with respect to a given PCTL by solving the \emph{symbolic solution function synthesis} problem \cite{daws_symbolic_2005,abdulla_model_2011,param,junges_parameter_2019}.

\begin{problem}[Symbolic Solution Function Synthesis]
    Given parameter set $\paramset$, a pDTMC $\mc_{\paramset}$ with rational state transitions, and a PCTL property, $\varphi$, find the rational function $f$ over parameters in $\paramset$ such that 
    \begin{equation}
        f_{\varphi}(\paramset) = \prsat{\mc_{\paramset}}.
        \label{eqn:symbsol}
    \end{equation}
    \label{prob:rational_func}
\end{problem}

We denote the instantiation of the symbolic solution function to be $\kappa(f_\varphi(\paramset))$, which indicates the probability that $\varphi$ is satisfied by the instance DTMC resulting from applying instantiation function $\kappa$ to each variable in $\mc_{\paramset}$.

\section{Adversarial Robustness in DTMCs}\label{sec:adv_rob}
In this section, we introduce a general definition of adversarial robustness for DTMCs and discuss in detail a class of $\varepsilon,d$-bounded threat models. We also present three DTMC adversarial robustness problems and extend these definitions to include deterministic, memoryless policies acting in MDPs. Additionally, we provide a reduction to the Uncertain Markov Chain (UMC) model checking problem \cite{hermanns_model-checking_2006} for $\varepsilon,max$-bounded threat models.

We choose to model our problem using the definitions and conventions of probabilistic model checking to enable the use of state-of-the-art probabilistic model checking tools such as PRISM \cite{prism} and PARAM \cite{param} to develop an efficient solution. We note, however that we may be able to model the same problem using other formalisms including process calculi or a formalism related to the one proposed in Topcu \etal \cite{topcu2012}. We leave this analysis to future work.

\subsection{Robustness Definition}\label{sec:gen_rob}
We introduce a general definition of \emph{DTMC adversarial robustness} with respect to a given property. In a real-world scenario, adversaries are restricted in how much they can alter a given system, both by physical limitations and in order to avoid detection. We model an adversary who can modify the system within some predetermined constraints.

To model this system perturbation, we introduce the \emph{perturbation space}, $\pertset{\mc}$ of DTMC $\mc=(\state,s_0,\probmat,AP,L)$ which consists of all DTMCs of the form $(\state,s_0,\probmat',AP,L)$. In other words, DTMCs in the perturbation space have the same states (with the same labeling and initial state) as the original DTMC, but vary in their transition probability matrix. We refer to subsets $PS\subseteq\pertset{\mc}$ as \emph{perturbation sets}. Any transition which has more than one possible value in $PS$ is considered a \emph{vulnerable transition}. All perturbation sets must include the original DTMC, and the minimal perturbation set for $\mc$ is $\{\mc\}$. These perturbation sets provide constraints on the adversary's ability to modify the transitions in the system. Using this notion of perturbation constraints, we can now define general adversarial robustness for a DTMC.

\begin{definition}[DTMC Adversarial Robustness]
    Given $0\leq\delta\leq1$, DTMC \mc, perturbation set $PS\subseteq\pertset{\mc}$, and some PCTL$^*$ path formula $\varphi$, \mc\ is \emph{adversarially robust} with respect to $PS,\;\varphi,\; \delta$ if
    \begin{equation}
        \prsat{\mc'} \geq \prsat{\mc} - \delta
        \label{eqn:advrob}
    \end{equation}
    for all $\mc'\in PS$.
    \label{def:advrob}
\end{definition}

We choose to define DTMC adversarial robustness as (\ref{eqn:advrob}) because it highlights the relationship between the satisfaction probability in the original DTMC and perturbed DTMC. It can also be written as a PCTL* satisfaction relation.
\begin{remark}
    Inequality (\ref{eqn:advrob}) is semantically equivalent to $s_0 \models_{\mc'} \mathbb{P}_{\geq (p - \delta)}(\varphi)$ where $p=\prsat{\mc}$ for all $\mc'\in PS$.
    \label{rem:rob_redefine_pctl}
\end{remark}

\subsection{Threat Models}\label{sec:tms}
The power of Definition \ref{def:advrob} lies in the ability to specify an adversarial perturbation set which represents an explicit threat model. There are many potential methods to restrict the adversary, for example, in a \emph{structure-preserving} threat model, the adversary is only able to perturb probabilities of transitions which exist in the original system. Formally this can be expressed as $\forall\mc'\in PS\subseteq\pertset{\mc}$, $\probmat_{s,s'}=0 \implies \probmat'(s,s') = 0$.

Definition \ref{def:advrob} begets a natural class of threat models for which the perturbation set is described as a ball around the probability matrix of the original DTMC. We will call this class the \emph{$\varepsilon,d$-bounded threat models}. An adversary acting under an $\varepsilon,d$-bounded threat model has a fixed budget, $\varepsilon$, which bounds the distance between the original and perturbed DTMCs, measured by a distance function $d$ over the transition probability matrices. Formally, we say $\forall\mc'\in PS\subseteq\pertset{\mc},\;d(\probmat,\probmat')\leq\varepsilon$. This $\varepsilon$ bound models the amount an attacker can modify the system. This value could represent physical limitations, for example if an attacker is able to move a shelf or barrel in the warehouse robot example from Sec. \ref{sec:intro}. It can also reflect the amount an attacker could modify the system while avoiding detection. For example in the collision avoidance protocol from Sec. \ref{sec:intro}, an attacker who reduces the probability of a message send by $90\%$ would result in noticeable slowdowns in transmission, likely alerting a system administrator that something is amiss.

Because of this modular framework, explicit threat models can be defined using combinations of these restrictions. This allows the system designer to evaluate many different potential attacks in a methodical manner. For the purposes of this paper, we focus our evaluation on threat models which are bounded with respect to the entry-wise $max$ distance defined as
\begin{equation}
    ||\pertmat - \pertmat'||_{max} = \max_{i,j}(|\pertmat_{ij}-\pertmat_{ij}'|).
\end{equation}
This norm is equivalent to the $\ell_\infty$ norm over the vector representation of matrix \pertmat. The following four perturbation set definitions are examples of these $\varepsilon,max$-bounded threat models.

\subsubsection{Selected transitions (ST) threat model}
The ST threat model permits adversaries to perturb only selected transitions. For instance, in the Grid World example from Sec. \ref{sec:intro}, an ST adversary who is able to remove certain types of obstacles (e.g., barrels), but not others (e.g., shelves) can introduce new transitions to the grid from a selected set. For this threat model, we specify a set $\mathcal{T}\subseteq\state\times\state$ of vulnerable transitions, and define the perturbation set to be
\begin{align}
    PS_{\varepsilon,max,\mathcal{T}} = \{ \mc' :\; &||\probmat - \probmat'||_{max} \leq \varepsilon \text{ and}\\
    &\nonumber  \forall (s,s') \not\in\mathcal{T},\; \probmat'_{s,s'}=\probmat_{s,s'}\}.
\end{align}

\subsubsection{Structure-preserving selected transitions (SPST) threat model}
The SPST threat model is a more restrictive version of the ST threat model. SPST adversaries may perturb transitions from a vulnerable set whose values in the original DTMC are non-zero. In other words, adversaries under this threat model may attack any transition in the vulnerable set, but may not add transitions to the DTMC. For example, an SPST attacker for the Grid World example from Sec. \ref{sec:intro} may increase the chance that an obstacle such as another robot is in specific spaces on the floor, decreasing the probability of moving into those spaces. For this threat model, we specify a set $\mathcal{T}\subseteq\state\times\state$ of vulnerable transitions, and define the perturbation set to be
\begin{align}
    PS_{\varepsilon,max,\mathcal{T},pres} = \{ \mc' :\; &||\probmat - \probmat'||_{max} \leq \varepsilon \text{ and}\\
    &\nonumber  \forall (s,s') \not\in\mathcal{T}\; \probmat'_{s,s'}=\probmat_{s,s'} \text{and } \\
    &\nonumber\forall s,s' \in \state,\;\probmat_{s,s'}=0 \implies \probmat'_{s,s'}=0\}.
\end{align}

\subsubsection{Selected states (SS) threat model}
A real-world adversary may be able to gain control of certain states in a system, and therefore it is useful to model an adversary by the states they can control. For example, an SS adversary for the collision avoidance protocol from Fig. \ref{fig:backoff} in Sec. \ref{sec:intro} with control of the receiver's ability to send an acknowledgment can increase the probability the receiver moves to RCV\_ACK. This increases the probability of a dropped message, which was not possible in the original protocol. The SS threat model is a special case of ST threat model, wherein the vulnerable transitions are defined to be those emanating from a set of vulnerable states. For this threat model, we specify a set $\mathcal{V}\subseteq\state$ of vulnerable states, and define the perturbation set to be
\begin{align}
    PS_{\varepsilon,max,\mathcal{V}} = \{ \mc' :\; &||\probmat - \probmat'||_{max} \leq \varepsilon \text{ and}\\
    &\nonumber  \forall s \not\in\mathcal{V},s'\in\state,\; \probmat'_{s,s'}=\probmat_{s,s'}\}.
\end{align}

\subsubsection{Structure-preserving selected states (SPSS) threat model}
The SPSS threat model is a special case of the SPST threat model. Here, adversaries are permitted to attack transitions emanating from a set of vulnerable states, but may not add transitions to the original DTMC. For example, an SPSS attacker for the collision avoidance protocol from Fig. \ref{fig:backoff} in Sec. \ref{sec:intro} with access to a sender may increase the probability that the sender waits to send a message, thereby causing a longer expected delay in message sends and unnecessarily slowing down communication. For this threat model, we specify a set $\mathcal{V}\subseteq\state$ of vulnerable states, and define the perturbation set to be
\begin{align}
    PS_{\varepsilon,max,\mathcal{V},pres} = \{ \mc' :\; &||\probmat - \probmat'||_{max} \leq \varepsilon \text{ and}\\
    &\nonumber  \forall s \not\in\mathcal{V},s'\in\state,\; \probmat'_{s,s'}=\probmat_{s,s'} \text{and } \\
    &\nonumber\forall s,s' \in \state,\;\probmat_{s,s'}=0 \implies \probmat'_{s,s'}=0\}.
\end{align}

\subsubsection{Other Threat Models}
We specify the previous four threat models to provide a starting point for formally modeling specific attackers in real-world scenarios. However, our framework is extremely flexible and is capable of modeling a wide variety of other potential attackers. For instance, we can consider a threat model in which the adversary has a cumulative budget that they can use across all transitions, using an $\varepsilon,\ell_2$-bounded threat model. The ability to add a bounded number of states not present in the original model can be accomplished by generating a surrogate DTMC with a set of additional, unconnected states which represent unknown states. We leave analysis of these additional threat models to future work.

\subsection{Verification and Synthesis Problems}\label{sec:probs}
We now present three main problem definitions related to verification of robustness and attack synthesis in DTMCs.

\begin{problem}[Verification of Adversarial Robustness]
    Given $0\leq\delta\leq1$, DTMC \mc{}, PCTL* path formula $\varphi$, $PS\subseteq\pertset{\mc}$, determine whether \mc\ is adversarially robust with respect to $PS,\;\varphi,\;\delta$.\label{prob:verif}
\end{problem}

This definition leads to two immediate results.
\begin{remark}[Zero Case]
    Adversarial robustness with an $\varepsilon,d$-bounded adversary always holds for $\varepsilon=0,\;\delta=0$.
\end{remark}
\begin{remark}[Monotonicity]
    If DTMC \mc\ is adversarially robust with an $\varepsilon,d$-bounded adversary, it is adversarially robust with respect to an equivalently defined $\varepsilon',d$-bounded adversary (varying only in the value of $\varepsilon$) for $\varepsilon'\leq\varepsilon$.
    \label{rem:mono}
\end{remark}

We provide the adversarial robustness definition with respect to $\delta$ so that the system designer can set an acceptable threshold for performance. These thresholds can be defined at an institutional or legal level according to specific requirements. For example, in the collision avoidance protocol from Sec. \ref{sec:intro}, a product manager might mandate that the acceptable probability that a message is received within 20 time ticks is at least $90\%$, chosen based on customer data indicating how long someone is willing to wait for a transaction to occur. The $\delta$ variable encodes the difference between the system behavior assuming no attack and the acceptable threshold.

It is also reasonable that the system designer may also wish to evaluate the maximum damage an adversary could do given an attack budget, to help determine a reasonable $\delta$ threshold for the system. For this scenario, we define the maximal $\delta$ synthesis problem.

\begin{problem}[Maximal $\delta$ Synthesis]
    Given DTMC \mc, PCTL* path formula $\varphi$, and perturbation set $PS\subseteq\pertset{\mc}$, synthesize the maximal $\delta$ for which \mc\ is adversarially robust with respect to $PS,\;\varphi,\;\delta$.\label{prob:dsynth}
\end{problem}

In order to solve Problems \ref{prob:verif} and \ref{prob:dsynth} we solve a more general synthesis problem for a worst-case attack. This attack provides meaningful information on the transition probabilities which have the most impact on robustness. This feedback can be used to determine which elements of a system are most important to maintain overall utility and are thereby the most critical to be secured.

\begin{problem}[Worst Case Attack Synthesis]
    Given DTMC $\mc=(\state,s_0,\probmat,AP,L)$, PCTL$^*$ path formula $\varphi$, and perturbation set $PS\subseteq\pertset{\mc}$, find perturbation matrix $\pertmat$ such that $\mc^*=(\state,s_0,\probmat+\pertmat,AP,L)$ satisfies $\prsat{\mc^*}\leq \prsat{\mc'} \; \forall\mc'\in PS$.\label{prob:advex}
\end{problem}
It follows immediately that Problems \ref{prob:verif} and \ref{prob:dsynth} reduce to Problem \ref{prob:advex} because the result of Problem \ref{prob:advex} is the attack perturbation which minimizes the probability of satisfying the property.

\begin{lemma}[Problem \ref{prob:verif} reduces to Problem \ref{prob:advex}]
    For $\mc^*$ which is the solution to Problem \ref{prob:advex} with given $\mc,\;\varphi,\;PS$,  $\prsat{\mc^*}\geq\prsat{\mc}-\delta$ for a given $0\leq\delta\leq 1$ if and only if \mc\ is adversarially robust with respect to $PS,\;\varphi,\;\delta$.
    \label{lem:advsolvesverif}
\end{lemma}

\begin{proof}
    We use the solution to Problem \ref{prob:advex} to compute $\prsat{\mc^*}$. If $\prsat{\mc^*}\geq\prsat{\mc}-\delta$, $\mc$ is adversarially robust with respect to $PS,\; \varphi,\; \delta$.
\end{proof}

\begin{lemma}[Problem \ref{prob:dsynth} reduces to Problem \ref{prob:advex}]
    For $\mc^*$ which is the solution to Problem \ref{prob:advex} with given $\mc,\;\varphi,\;PS$,
    \begin{equation}
        \delta^*::=\prsat{\mc}-\prsat{\mc^*}
    \end{equation}
    is the maximal $\delta$ for which \mc\ is adversarially robust with respect to $PS,\;\varphi,\;\delta$.
    \label{lem:advsolvesmindelta}
\end{lemma}

\begin{proof}
    We use the solution from Problem \ref{prob:advex} to compute $\prsat{\mc^*}$. Then we have that $\delta^*::=\prsat{\mc}-\prsat{\mc^*}$ is the solution to Problem \ref{prob:dsynth}.
\end{proof}

\subsection{Reduction to the UMC Model Checking Problem}\label{sec:umc}
The uncertain Markov chain (UMC) model checking problem is to determine whether all DTMCs in an infinite set defined by an interval on the entries of the transition probability matrix satisfy a PCTL state formula~\cite{hermanns_model-checking_2006}. In this section we show that one of our three problems (DTMC adversarial robustness verification, Problem \ref{prob:verif}) reduces to the uncertain Markov chain (UMC) model checking problem for $\varepsilon,max$-bounded threat models and PCTL properties. We first revisit background on interval-valued Markov chains and the UMC model checking problem~\cite{hermanns_model-checking_2006}.
\begin{definition}[Interval-Valued Discrete Time Markov Chain (IDTMC)]
    An \emph{interval-valued discrete time Markov chain} is a tuple 
    \begin{equation}
        \mathcal{I}=(\state,s_0,\lbprobmat,\ubprobmat,AP,L)   
    \end{equation}
    where $\state$ is a finite set of states with $s_0$ initial state. $\lbprobmat$ is a transition probability matrix for which $\lbprobmat_{s,s'}$ gives the lower bound of the transition probability from state $s$ to state $s'$. $\ubprobmat$ similarly is the upper bound transition probability matrix. $AP$ is a set of atomic propositions and $L:\state\rightarrow2^{AP}$ is a labeling function.
\end{definition}
  $[\mathcal{I}]$ is the infinite set of DTMCs $(\state,s_0,\probmat,AP,L)$ such that $\lbprobmat_{s,s'}\leq\probmat_{s,s'}\leq\ubprobmat_{s,s'}\;\forall s,s'\in\state$. For the uncertain Markov chain semantics of IDTMC $\mathcal{I}$, we assume that an external environment non-deterministically picks a single DTMC from $[\mathcal{I}]$ at the beginning. Now we revisit the UMC model checking problem~\cite{hermanns_model-checking_2006}.
\begin{problem}[UMC Model Checking Problem]
    Given a IDTMC $\mathcal{I}$ and PCTL state formula $\Phi$, $s_0\models_{\mathcal{I}} \Phi$ if and only if, $\forall\mc\in[\mathcal{I}]$, $s_0\models_{\mc}\Phi$.
    \label{prob:UMC}
\end{problem}
In other words, $[\mathcal{I}]$ does not satisfy $\Phi$ if there exists a DTMC $\mc$ in $[\mathcal{I}]$ such that $\mc$ does not satisfy $\Phi$. 

We now examine how the adversarial robustness verification problem (Problem \ref{prob:verif}) relates to the UMC model checking problem (Problem \ref{prob:UMC}).
\begin{lemma}
    Given DTMC \mc, $\varepsilon,max$-bounded perturbation set $PS$, and $0\leq\delta\leq 1$, we can define IDTMC $\mathcal{I}$, such that, for the DTMC \mc, 
    \begin{equation}
        \mc\in[\mathcal{I}]\iff\mc\in PS
    \end{equation}
    \label{lem:idtmc_ps}
\end{lemma}
\begin{proof}
    Assume we have DTMC $\mc=(\state,s_0,\probmat,AP,L)$, $0\leq\delta\leq1$, and $\varepsilon,max$-bounded perturbation set $PS\subseteq\pertset{\mc}$ with $0\leq\varepsilon\leq1$ and set of vulnerable transitions $\mathcal{V}\subseteq\state\times\state$. We then construct IDTMC $\mathcal{I}=(\state,s_0,\lbprobmat,\ubprobmat,AP,L)$ where $\ubprobmat_{s,s'}=\probmat_{s,s'}+min(\varepsilon,1-\probmat_{s,s'})$ and $\lbprobmat_{s,s'}=\probmat_{s,s'}-min(\varepsilon,\probmat_{s,s'})$ for all $(s,s')\in\mathcal{V}$, and $\ubprobmat_{s,s'}=\lbprobmat_{s,s'}=\probmat_{s,s'}$ for all $(s,s')\not\in\mathcal{V}$. Consider $\mc'=(\state,s_0,\probmat',AP,L)\in[\mathcal{I}^*]$. By definition, $||\probmat-\probmat'||_{max}\leq\varepsilon$ and  $\probmat'=\probmat_{s,s'}=\probmat_{s,s'}\;\forall (s,s')\not\in\mathcal{V}$. By definition of $\mathcal{I}^*$, we also have that $\forall s\in\state,\;\sum_{s'\in\state}\probmat'_{s,s'}=1$ and $\forall s,s'\in\state,\;0\leq\probmat'_{s,s'}\leq 1$. By the definition of the $\varepsilon,max$-bounded threat model, we have that $\mc'\in PS$
    
    Now consider  $\mc'=(\state,s_0,\probmat'=\probmat+\pertmat,AP,L)\in PS$. 
    Since $PS\subseteq\pertset{\mc}$, we know that $\forall s\in\state,\;\sum_{s'\in\state}\probmat'_{s,s'}=1$.
    $PS\subseteq\pertset{\mc}$ also implies that $0\leq\probmat'_{s,s'}\leq 1\;\forall s,s'\in\state$. Because $\mc'\in PS$ where $PS$ is an $\varepsilon,max$-bounded perturbation matrix, $\pertmat'$ is a valid perturbation matrix. We also have that $\lbprobmat_{s,s'}\leq\probmat'_{s,s'}
    \leq\ubprobmat_{s,s'}\; \forall s,s'\in\state$. By definition of $\mathcal{I}^*$ we have that $\mc'\in[\mathcal{I}^*]$
\end{proof}
\begin{thm}
    Verification of DTMC adversarial robustness for an $\varepsilon,max$-bounded adversary and a PCTL path formula, reduces to the Uncertain Markov Chain model checking problem.
    \label{thm:umc_reduction}
\end{thm}
\begin{proof}
    Assume we have $0\leq\delta\leq1$, DTMC $\mc=(\state,s_0,\probmat,AP,L)$, PCTL path formula $\varphi$, and $\varepsilon,max$-bounded threat model. Assume also we have algorithm $A$ which takes as input IDTMC $\mathcal{I}$ and PCTL state formula $\Phi$ and returns $s_0\models_\mathcal{I}\Phi$.
    
    Using Lemma \ref{lem:idtmc_ps}, we define IDTMC $\mathcal{I}^*$ such that $\mc\in[\mathcal{I}^*]\iff\mc\in PS$. We also define PCTL state formula $\Phi^*=\mathbb{P}_{\geq (p - \delta)}(\varphi)$ where $p=\prsat{\mc}$. We run algorithm $A$ on $\mathcal{I}^*,\Phi^*$. If $A(\mathcal{I^*},\Phi^*)=true$, we know that $\nexists \mc \in [\mathcal{I}^*]$ such that $s_0\models_{\mc}\Phi^*$. If $A(\mathcal{I^*},\Phi^*)=false$, we know that $\exists \mc \in [\mathcal{I}^*]$ such that $s_0\models_{\mc}\Phi^*$.
    Therefore, by Remark \ref{rem:rob_redefine_pctl} and Lemma \ref{lem:idtmc_ps}, we have that \mc\ is adversarially robust with respect to $PS,\;\varphi,\;\delta$ if and only if $A(\mathcal{I^*},\Phi^*)$.
\end{proof}

\subsection{Policy Robustness in MDPs}\label{sec:DTMCtoMDP}
Our definition of adversarial robustness for DTMCs is easily extended to agents acting in MDPs under deterministic, memoryless policies. In this case, the adversary is able to perturb the transition probabilities in the underlying MDP, with the goal of decreasing the satisfaction probability of the DTMC formed by composing the policy and MDP. In this case we define the perturbation space, $\pertset{\mdp}$ of MDP $\mdp=(\state,s_0,\action,\transition,AP,L)$ to be the set of all MDPs of the form $(\state,s_0,\action,\transition{}',AP,L)$. This extension is useful when considering systems modeled as agents acting in probabilistic environments, as is the case for reinforcement learning and our warehouse robot example from Sec. \ref{sec:intro}.

\begin{definition}[Policy Adversarial Robustness]
    Given $0\leq\delta\leq1$, MDP  \mdp{}, perturbation set $PS\subseteq\pertset{\mdp}$, and some PCTL* path formula $\varphi$, a policy $\sigma$ is \emph{adversarially robust} in \mdp{} with respect to $PS,\;\varphi,\;\delta$ if, for the DTMC $\mdp_\sigma$,
    \begin{equation}
        \prsat{\mdp'_{\sigma}} \geq \prsat{\mdp_{\sigma}} - \delta
    \end{equation}
    for all $\mdp'\in PS$.
\end{definition}

We similarly define the three DTMC adversarial robustness problems from Section \ref{sec:probs} for policies acting in MDPs. Though the adversary may modify any vulnerable transition in the MDP, the adversary should only perturb the transition probabilities which correspond to the possible transitions in the resulting DTMC. It is enough to consider the state transition probabilities in the underlying MDP which are present in the composed DTMC. Therefore, \pertset{\mdp} and \pertset{\mdp_\sigma} are equivalent with respect to policy adversarial robustness for given $\mdp,\;\sigma$.

\begin{lemma}[Robustly Equivalent MDPs]
    Given a deterministic, memoryless policy $\sigma$, $\mdp=(\state,s_0,T,AP,L)$, perturbation set $PS\subseteq\pertset{\mdp}$, PCTL$^*$ path formula $\varphi$, and $\mdp^*=(\state,s_0,T^*,AP,L)$ such that $T(s,a,s')=T^*(s,a,s')$ for all $a\in\action,s\in\state$ where $\sigma(s)=a$, we have that $\sigma$ is adversarially robust in $\mdp$ with respect to $PS,\;\varphi,\;\delta$ if and only if it is adversarially robust in $\mdp^*$ with respect to $PS,\;\varphi\;\delta$.
    \label{lem:DTMC_trans}
\end{lemma}
\begin{proof}
    Assume we have deterministic, memoryless policy $\sigma$, $\mdp=(\state,s_0,T,AP,L)$, perturbation set $PS\subseteq\pertset{\mdp}$, PCTL* path formula $\varphi$, and $\mdp^*=(\state,s_0,T^*,AP,L)$ such that $T(s,a,s')=T^*(s,a,s')$ for all $a\in\action,s\in\state$ where $\sigma(s)=a$. By Definition \ref{def:DTMCofMDP}, $\mdp^*_\sigma=\mdp_\sigma$. Therefore, we have that $\sigma$ is robust in $\mdp$ with respect to $PS,\;\varphi,\;\delta$ if and only if it is robust in $\mdp^*$ with respect to $PS,\;\varphi,\;\delta$.
\end{proof}

It immediately follows that the three adversarial robustness problems for policies acting in MDPs reduce to the three adversarial robustness problems for DTMCs, respectively.

We note that this definition could also be extended such that the attacker can modify the policy, rather than the transition probabilities in the underlying MDP. However, this would be a discrete modification, changing from one action in a finite set to another. In the composed DTMC, this would result in a similarly discrete modification from one set of state transition probabilities to another, rather than a continuous change as is the case when you hold the policy constant. We choose to focus on the more interesting continuous scenario for the purposes of this paper.
\section{Optimization Solutions to the DTMC Adversarial Robustness Problems}\label{sec:solns}
In the following subsections, we reformulate the worst-case attack synthesis problem (Problem \ref{prob:advex}) as an optimization problem and provide two methods to solve the optimization using existing probabilistic model checking tools. We define these methods as \emph{optimization with direct computation} and \emph{optimization with symbolic solution function}. Due to the reductions in Lemmas \ref{lem:advsolvesverif} and \ref{lem:advsolvesmindelta}, these methods will additionally solve the other two DTMC robustness problems (Problems \ref{prob:verif} and \ref{prob:dsynth}). In Section \ref{sec:eval} we evaluate and compare the two approaches over case studies of varying state size and number of parameters.

\subsection{Perturbation Matrix Synthesis as an Optimization}\label{sec:opt_sol}
We present a constrained optimization over the perturbation matrix $\pertmat$ to find the worst case attack.
\begin{problem}[Perturbation Matrix Synthesis]
    Given DTMC $\mc=(\state,s_0,\probmat,AP,L)$, perturbation set $PS\subseteq\pertset{\mc}$, and some PCTL$^*$ path formula $\varphi$, find a perturbation matrix $\pertmat^*$ which solves
    \begin{align}
        \mathop{argmin}_{\pertmat}\;&\prsat{\mc'}\label{eqn:objective} \\
        \nonumber subject\; to \;& \mc'=(\state,s_0,\probmat+\pertmat,AP,L)\in PS
    \end{align}
    \label{prob:opt}
\end{problem}

\begin{lemma}
    Given DTMC $\mc=(\state,s_0,\probmat,AP,L)$, perturbation set $PS\subseteq\pertset{\mc}$, and some PCTL$^*$ path formula $\varphi$, $\pertmat^*$ which solves Problem \ref{prob:opt} is the worst case attack for \mc\ with respect to $PS,\;\varphi$.
\end{lemma}
\begin{proof}
    Assume DTMC $\mc$, perturbation set $PS\subseteq\pertset{\mc}$, and PCTL$^*$ path formula $\varphi$. By definition, $\pertmat^*$ which solves Problem \ref{prob:opt} for \mc, $PS$, and $\varphi$ is the worst case attack for \mc\ with respect to $PS,\;\varphi$.
\end{proof}

It follows from Lemmas \ref{lem:advsolvesverif} and \ref{lem:advsolvesmindelta} that the solution to the attack synthesis optimization can also be used to solve Problems \ref{prob:verif} and \ref{prob:dsynth}.

We note that the form of the objective function (\ref{eqn:objective}) is dependent on the property $\varphi$. For example, the objective is linear with respect to next step properties (reachability in exactly one time step), quadratic with respect to next next step properties (reachability in exactly two time steps), and is non-convex in general. Similarly, the form of the constraints are dependent on the threat model. For example, the constraints corresponding to the $\varepsilon,max$-bounded threat models described in Section \ref{sec:tms} are linear. This means that the techniques available for solving this optimization can change based on the property and threat model.

\subsection{Solutions to the DTMC Adversarial Robustness Problems for $\varepsilon,max$-bounded threat models}\label{sec:objective}
We propose two solutions to the optimization problem for $\varepsilon,max$-bounded threat models: \emph{optimization with direct computation} and \emph{optimization with symbolic solution function}. Because the constraints are linear, but the shape of the objective depends on the property, both solutions use the Sequential Least Squares Programming (SLSQP) \cite{2020SciPy-NMeth} algorithm to search the parameter space for the minimal value. The SLSQP algorithm performs an iterative search of the space to find a minimum. The main difference between the two solutions is the technique used to compute the objective at each iteration. The first method performs a model-checking procedure at every function call. The second method constructs a pDTMC from the original DTMC before the optimization begins, and then simply instantiates this function at each call to the objective.

\subsubsection{Optimization with Direct Computation}
In the optimization with direct computation solution, we solve the optimization presented in Problem \ref{prob:opt} using the SLSQP iterative algorithm. At each iteration, we construct the perturbed DTMC, $\mc'=(\state,s_0,\probmat,AP,L)$, and use existing model checking techniques to compute $\prsat{\mc'}$ directly. This solution is beneficial because computing $\prsat{\mc'}$ for a DTMC is solved for many logics including PCTL and PCTL$^*$. Tools such as Prism \cite{prism} and Storm \cite{storm} compute these probabilities efficiently for large state spaces.

\subsubsection{Optimization with Symbolic Solution Function}
In the optimization with symbolic solution function approach, we construct a pDTMC $\mc_X$ from the original DTMC $\mc$, replacing vulnerable transitions with parameters. In some cases, like when considering PCTL properties, we can then find the symbolic solution function for $\prsat{\mc_X}$ (as described in eqn. (\ref{eqn:symbsol})). We then run the SLSQP algorithm for searching the parameter space to find the minimum. This eliminates the need to perform model checking at each iteration. In Algorithm \ref{alg:symb}, we describe this procedure in detail.

\begin{algorithm}
    \caption{Symbolic Objective Approach}
    \label{alg:symb}
    \begin{algorithmic}[1]
    \renewcommand{\algorithmicrequire}{\textbf{Input:}}
    \renewcommand{\algorithmicensure}{\textbf{Output:}}
    \REQUIRE DTMC $\mc=(\state,s_0,\probmat,AP,L)$, $PS\subseteq\pertset{\mc}$, PCTL path formula $\varphi$
    \ENSURE perturbation matrix $\pertmat^*$
    \STATE $X \gets O$, $\probmat^\paramset\gets \probmat$
    \FOR{$(s,s')$ which is a vulnerable transition in $PS$}
        \STATE add $v_{(s,s')}$ to \paramset\
        \STATE $\probmat^\paramset\gets v_{(s,s')}$ 
    \ENDFOR
    \STATE $\mc_\paramset\gets (\state, s_0, \probmat^\paramset, AP, L)$
    \STATE generate symbolic solution function $f_{\varphi}(X)$ for $\mc_\paramset$
    \STATE $\kappa^*\gets \mathop{argmin}_\kappa$ $\kappa(f_{\varphi}(\paramset))$ subject to $\kappa(\mc_\paramset)\in PS$
    \RETURN $\kappa^*(\probmat^\paramset)-\probmat$ 
    \end{algorithmic} 
    \end{algorithm}

State-of-the-art parametric model checking tools such as PARAM~\cite{param}, Storm~\cite{storm}, PROPhESY~\cite{prophesy}, and PRISM~\cite{prism} can generate the symbolic solution function efficiently for few parameters on subsets of PCTL. Generating this function can become expensive as the number of parameters in the constructed pDTMC (and complexity of the symbolic function) increases. This method therefore becomes less desirable under less restrictive threat models with many vulnerable transitions. Additionally, there is no known process for generating symbolic solution functions for all PCTL$^*$ properties. We perform an in-depth comparison of these two methods over case studies of various sizes with different threat models in Section \ref{sec:evalopts}.

\section{Case Studies and Evaluation}\label{sec:eval}
In the following sections we evaluate the two optimization solutions to our three adversarial robustness problems, as described in Section \ref{sec:opt_sol}. We implement both the direct computation method and the symbolic solution function approach. We compare the two approaches over different parameter and state sizes. We then use these solutions to demonstrate the flexibility and utility of our definitions over three case studies. Code for all experiments can be found here: \url{https://github.com/lisaoakley/dtmc_attack_synthesis}.

\subsection{Case Studies}
We perform our evaluation on two protocol DTMCs, and one DTMC defined by an agent acting in an MDP.

\subsubsection{Simple Communication Protocol}
We first show our technique works on the simple communication protocol from \cite{baier2008principles}, described in Fig. \ref{fig:simpleprot}. In this protocol, a sent message can be delivered or lost. When it is lost, the sender re-sends the message infinitely many times until it is delivered. The goal of the system is to deliver the message quickly. In this system, we model a structure-preserving selected states (SPSS) adversary who can attack the ``try'' state by modifying the probability that a message is lost or delivered.
\begin{figure}[ht]
    \centering
    \centering
\scalebox{.7}{%
\begin{tikzpicture}[node distance = 3cm, initial text =,elliptic state/.style={draw,ellipse}]] 
    \node[elliptic state,initial,initial where=above] (start) {$start$}; 
    \node[elliptic state] (try) [below of=start] {$try$};
    \node[elliptic state] (lost) [right of=try] {$lost$};
    \node[elliptic state] (delivered) [left of=try] {$delivered$};

    % start to try
    \draw[->, thick] (start) to node [right,midway,rotate=0] {\small 1} node [swap] {} (try);

    % try to lost
    \draw[->, thick] (try) [out=-45,in=180+45] to node [above,midway,rotate=0] {\small .2} node [swap] {} (lost);
    % try to delivered
    \draw[->, thick] (try) to node [above,midway,rotate=0] {\small .8} node [swap] {} (delivered);

    % lost to try
    \draw[->, thick] (lost) [out=180-45,in=45] to node [above,midway,rotate=0] {\small 1} node [swap] {} (try);

    % delivered to start
    \draw[->, thick] (delivered) to node [above,midway] {\small 1} node [swap] {} (start);
\end{tikzpicture} 
}
    \caption{Simple Communication Protocol from \cite{baier2008principles}.}
    \label{fig:simpleprot}
\end{figure}
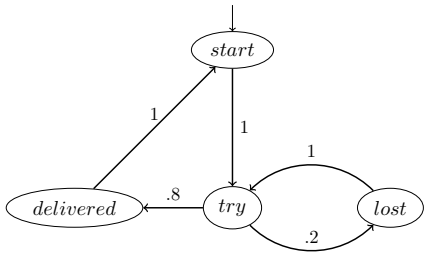

\subsubsection{IPv4 Zeroconf Protocol}

The IPv4 Zeroconf Protocol is a network protocol that is commonly used to evaluate model-checking approaches \cite{hermanns_model-checking_2006}. Here we consider the DTMC representing the procedure for a new host joining the network as described in Fig. \ref{fig:zeroconf}. In this procedure, we assume a fixed number of $m$ hosts already on the network. When a new host joins the network, they immediately generate a random address from a predetermined set of $K=65024$ possible addresses and broadcasts it to the existing hosts. With probability $m/K$, this address is a collision. The host then waits a predetermined $n$ time ticks to get a response from the network. Assuming there is a collision, at each time tick there is some probability $p$ that the new host learns of this mistake and is able to re-generate their address. The goal of the host is to successfully join the network with a unique address, in as little time as possible. 

In this case, the adversary attempts to prevent existing hosts from notifying the new host of the collision. We model this attack by a structure-preserving selected states (SPSS) threat model, allowing the adversary to reduce the probability at certain time ticks (states in $\{1,\dots,n\}$). In the case of the symbolic approach, this is equivalent to re-building the DTMC with unique probabilities for transitions between time ticks in the vulnerable set. In prior work on pDTMCs, the IPv4 protocol was analyzed by varying a single parameter $p$ \cite{hermanns_model-checking_2006}. Our evaluation of this case study is more expansive, allowing the adversary to attack some, but not all of the transition probabilities.

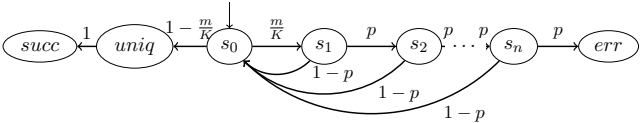
\begin{figure}[ht]
    \centering
    \centering
\scalebox{.7}{%
\begin{tikzpicture}[node distance = 1.8cm, initial text =,elliptic state/.style={draw,ellipse}]] 
    \node[elliptic state,initial,initial where=above] (s0) {$s_0$}; 
    \node[elliptic state] (s1) [right of=s0] {$s_1$};
    \node[elliptic state] (s2) [right of=s1] {$s_2$};
    \node[elliptic state] (sn) [right of=s2] {$s_n$};
    \node[elliptic state] (err) [right of=sn] {$err$};
    \node[elliptic state] (unique) [left of=s0] {$uniq$};
    \node[elliptic state] (success) [left of=unique] {$succ$};

    % s0 to unique
    \draw[->, thick] (s0) to node [above,midway,rotate=0] {\small $1-\frac{m}{K}$} node [swap] {} (unique);
    % s0 to s1
    \draw[->, thick] (s0) to node [above,midway,rotate=0] {\small $\frac{m}{K}$} node [swap] {} (s1);

    % s1 to s2
    \draw[->, thick] (s1) to node [above,midway,rotate=0] {\small $p$} node [swap] {} (s2);
    % s1 to s0
    \draw[->, thick] (s1) [out=180+45,in=-45] to node [right,xshift=15pt,midway,rotate=0] {\small $1-p$} node [swap] {} (s0);
    
    % \path (M) to coordinate[pos=0.3] (aux-1) coordinate[pos=0.7] (aux-2) (F);
    % \draw[-latex]   (M) -- (aux-1)
    %             (aux-2)-- (F);
    % \draw[dotted]    (aux-1) to ["rho"] (aux-2);
    
    % s2 to sn
    % \path (s2) to coordinate[pos=0.3] (aux-1) coordinate[pos=0.7] (aux-2) (sn)
    \draw[->, thick] (s2) to node [above,midway,rotate=0] {\small $p\;\;\;\;\;p$} node [pos=0.5,fill=white] {\dots} (sn);
    % s2 to s0
    \draw[->, thick] (s2) [out=180+45,in=-45] to node [right,xshift=25pt,midway,rotate=0] {\small $1-p$} node [swap] {} (s0);

    % sn to err
    \draw[->, thick] (sn) to node [above,midway,rotate=0] {\small $p$} node [swap] {} (err);
    % sn to s0
    \draw[->, thick] (sn) [out=180+45,in=-45] to node [right,xshift=35pt,midway,rotate=0] {\small $1-p$} node [swap] {} (s0);
    
    % unique to succ
    \draw[->, thick] (unique) to node [above,midway,rotate=0] {\small 1} node [swap] {} (success);

    % % try to lost
    % \draw[->, thick] (try) [out=-45,in=180+45] to node [above,midway,rotate=0] {\small .2} node [swap] {} (lost);
    % % try to delivered
    % \draw[->, thick] (try) to node [above,midway,rotate=0] {\small .8} node [swap] {} (delivered);

    % % lost to try
    % \draw[->, thick] (lost) [out=180-45,in=45] to node [above,midway,rotate=0] {\small 1} node [swap] {} (try);

    % % delivered to start
    % \draw[->, thick] (delivered) to node [above,midway] {\small 1} node [swap] {} (start);
\end{tikzpicture} 
}
    \caption{IPv4 Zeroconf Protocol \cite{hermanns_model-checking_2006}. Our model is a DTMC with initial state $s_0$ and states $s_1$ through $s_n$ indicating the number of time ticks the new host has waited to hear from existing hosts regarding a collision. State $err$ indicates the new host erroneously chooses a colliding address. In state $uniq$ the host has found a unique address and in state $succ$ the host has used that address to successfully join the network. }
    \label{fig:zeroconf}
\end{figure}

\begin{table*}[t]
    \centering
    \begin{tabular}{llllr}
        \toprule
                   Property & \# States & \# Params &               Method &            Total Duration (in seconds) \\
        \midrule
          P=? [s!=5 U s=24] &       25 &        5 &   Direct Computation &                     \textbf{0.354} \\
          P=? [s!=5 U s=24] &       25 &        5 & Symbolic Soln. Func. &     0.052 + 0.076 = \textbf{0.128} \\
          P=? [s!=5 U s=24] &       25 &       10 &   Direct Computation &                     \textbf{2.445} \\
          P=? [s!=5 U s=24] &       25 &       10 & Symbolic Soln. Func. &  124.316 + 0.98 = \textbf{125.296} \\
          P=? [s!=5 U s=24] &       25 &       20 &   Direct Computation &                     \textbf{4.884} \\
          P=? [s!=5 U s=24] &       25 &       20 & Symbolic Soln. Func. &                        Time out \\
         P=? [s!=10 U s=99] &      100 &        5 &   Direct Computation &                     \textbf{12.19} \\
         P=? [s!=10 U s=99] &      100 &        5 & Symbolic Soln. Func. &     2.006 + 1.089 = \textbf{3.095} \\
         P=? [s!=10 U s=99] &      100 &       10 &   Direct Computation &                     \textbf{23.69} \\
         P=? [s!=10 U s=99] &      100 &       10 & Symbolic Soln. Func. & 717.126 + 2.792 = \textbf{719.918} \\
         P=? [s!=10 U s=99] &      100 &       20 &   Direct Computation &                    \textbf{67.616} \\
         P=? [s!=10 U s=99] &      100 &       20 & Symbolic Soln. Func. &                        Time out \\
        P=? [s!=15 U s=224] &      225 &        5 &   Direct Computation &                     \textbf{59.32} \\
        P=? [s!=15 U s=224] &      225 &        5 & Symbolic Soln. Func. &                        Time out \\
        P=? [s!=15 U s=224] &      225 &       10 &   Direct Computation &                   \textbf{134.714} \\
        P=? [s!=15 U s=224] &      225 &       10 & Symbolic Soln. Func. &                        Time out \\
        P=? [s!=15 U s=224] &      225 &       20 &   Direct Computation &                   \textbf{290.201} \\
        P=? [s!=15 U s=224] &      225 &       20 & Symbolic Soln. Func. &                        Time out \\
        \bottomrule
        \end{tabular}
    \caption{Comparison between objective computation methods for $5\times 5$, $10\times 10$, and $15\times 15$ Grid World DTMCs with $5,\; 10,$ and $20$ vulnerable transitions (parameters). For symbolic solution, results are presented as ``symbolic solution function generation duration'' + ``optimization duration''. Pre-computing the symbolic solution function improves the optimization run-time, but the pre-computation step quickly becomes unmanageable as the number of parameters increases.}
    \label{tab:scale}
\end{table*}
\subsubsection{Grid World}

We consider the $n\times m$ Grid World~\cite{gridworld} class of case studies in which an agent navigates an $n\times m$ grid by moving up, down, left, and right according to a memoryless, deterministic policy. The Grid World environment is modeled as an MDP. We assume the policy is fixed, and use the policy adversarial robustness definition from Section \ref{sec:DTMCtoMDP} to reason about the composed DTMC. In each state of the composed DTMC, there is some probability that the agent does not go in their intended direction. We consider an adversary who is able to increase and decrease these probabilities, within constraints defined by an $\varepsilon,max$-bounded threat model. Adversaries to this model can use any of the four threat models defined in Section \ref{sec:tms}.

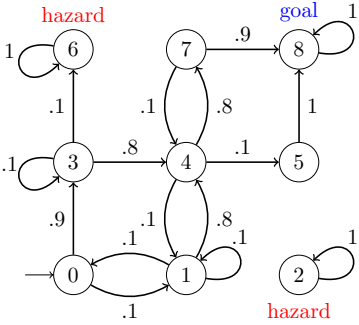
\begin{figure}[ht]
    \centering
    \centering
    \scalebox{.75}{%
    \begin{tikzpicture}[main/.style = {draw, circle},node distance = 2cm, initial text =,state/.style={circle, draw, minimum size=.7cm}]] 
        \node[state,initial] (0) {$0$}; 
        \node[state] (1) [right of=0] {$1$};
        \node[state] (2) [right of=1,label=below:{\normalsize \textcolor{red}{hazard}}] {$2$};
        \node[state] (3) [above of=0] {$3$};
        \node[state] (4) [above of=1] {$4$};
        \node[state] (5) [above of=2] {$5$};
        \node[state] (6) [above of=3,label={\normalsize \textcolor{red}{hazard}}] {$6$};
        \node[state] (7) [above of=4] {$7$};
        \node[state] (8) [above of=5,label={\normalsize \textcolor{blue}{goal}}] {$8$};

        % 0 right
        \draw[->, thick] (0) [out=-30,in=180+30] to node [below,midway,rotate=0] {\normalsize .1} node [swap] {} (1);
        % 0, up
        \draw[->, thick] (0)  to node [left, midway] {\normalsize .9} node [swap] {} (3);

        % 1 left
        \draw[->, thick] (1) [out=180-30,in=30] to node [above,midway,rotate=0] {\normalsize .1} node [swap] {} (0);
        % 1 self
        \draw[->, thick] (1) [out=-10,in=45,loop] to node [above,midway, yshift=3pt] {\normalsize .1} node [swap] {} (1);
        % 1 up
        \draw[->, thick] (1) [out=60,in=-60] to node [right,midway,rotate=0] {\normalsize .8} node [swap] {} (4);

        % 2 self
        \draw[->, thick] (2) [out=-10,in=45,loop] to node [above,midway, yshift=3pt] {\normalsize 1} node [swap] {} (2);

        % 3 self
        \draw[->, thick] (3) [out=180-10,in=180+45,loop] to node [above,midway, xshift=-5pt] {\normalsize .1} node [swap] {} (3);
        % 3 right
        \draw[->, thick] (3)  to node [above, midway] {\normalsize .8} node [swap] {} (4);
        % 3 up
        \draw[->, thick] (3)  to node [left, midway] {\normalsize .1} node [swap] {} (6);

        % 4 down
        \draw[->, thick] (4) [out=180+60,in=180-60] to node [left,midway,rotate=0] {\normalsize .1} node [swap] {} (1);
        % 4 right
        \draw[->, thick] (4)  to node [above, midway] {\normalsize .1} node [swap] {} (5);
        % 4 up
        \draw[->, thick] (4) [out=60,in=-60] to node [right,midway,rotate=0] {\normalsize .8} node [swap] {} (7);

        % 5 up
        \draw[->, thick] (5)  to node [right, midway] {\normalsize 1} node [swap] {} (8);

        % 6 self
        \draw[->, thick] (6) [out=180-10,in=180+45,loop] to node [above,midway, xshift=-5pt] {\normalsize 1} node [swap] {} (6);

        % 7 down
        \draw[->, thick] (7) [out=180+60,in=180-60] to node [left,midway,rotate=0] {\normalsize .1} node [swap] {} (4);
        % 7 right
        \draw[->, thick] (7)  to node [above, midway] {\normalsize .9} node [swap] {} (8);

        % 8 self
        \draw[->, thick] (8) [out=-10,in=45,loop] to node [above,midway, yshift=3pt] {\normalsize 1} node [swap] {} (8);
    \end{tikzpicture} 
    }
    \caption{Example $3\times 3$ Grid World DTMC \cite{gridworld}. Note that hazard state 2 is unreachable. Under a threat model that does not preserve structure, this hazard state can become reachable and affect the performance of the agent.}
    \label{gridworld}
\end{figure}
States in Grid World grids are indexed sequentially from left to right then bottom to top. For some experiments, we generate randomized $n\times m$ Grid World examples for a given $n,m,$ and transition structure. We accomplish this by generating a random probability distribution over adjacent states in the transition structure, for each state in the DTMC. In other words, the structure of each DTMC is fixed, but the probabilities are randomized.

\subsection{Experiment Setup}
For our optimizations, we use the SciPy optimize package~\cite{2020SciPy-NMeth} with the Sequential Least Squares Programming (SLSQP) algorithm. When directly computing \prsat{\mc}, we use the PRISM Java API and the Py4J python package~\cite{prism,py4j}. When computing the symbolic solution function, we use the PARAM parametric probabilistic model checker~\cite{param}. All experiments were run on a computer with an Intel Xeon Silver 4114 2.20GHz processor and 187GB RAM.

\subsection{Comparing Solutions}\label{sec:evalopts}
In Section \ref{sec:objective}, we presented two optimization solutions for the worst-case attack synthesis problem. In Table \ref{tab:scale}, we present results comparing the two techniques over randomized $5\times 5$, $10\times 10$, and $15\times 15$ Grid World DTMCs. We use an ST threat model with $5,\; 10$, and $20$ vulnerable transitions (parameters). We choose the state sizes and number of parameters to illustrate the relative growth in computation time between the two solutions over a large range of settings. We use a lower bound of 5 parameters, as this is close to the highest number of parameters typically considered in prior work on computing symbolic solution functions, and a upper bound of 20 parameters as the symbolic solution times out for larger values. We see that the symbolic solution improves the optimization time when compared to direct computation. As the number of parameters increases, however, the symbolic solution function quickly becomes expensive to generate, especially with larger state spaces. For example, in a $25$ and $100$ state Grid World with $5$ parameters, the symbolic computation is significantly faster. Once the number of parameters is greater than or equal to 10, the direct computation is always faster. For $225$ state Grid World and more than 10 parameters, the symbolic solution times out. Therefore, in cases of many vulnerable transitions, the direct computation method is preferred.

\begin{figure*}[ht]
    \centering
    \subfloat[Simple communication protocol satisfaction probabilities under attack with $\varphi=\lozenge^{\leq 10}delivered$.\label{simple}]{%
    \includegraphics[width=.4\linewidth]{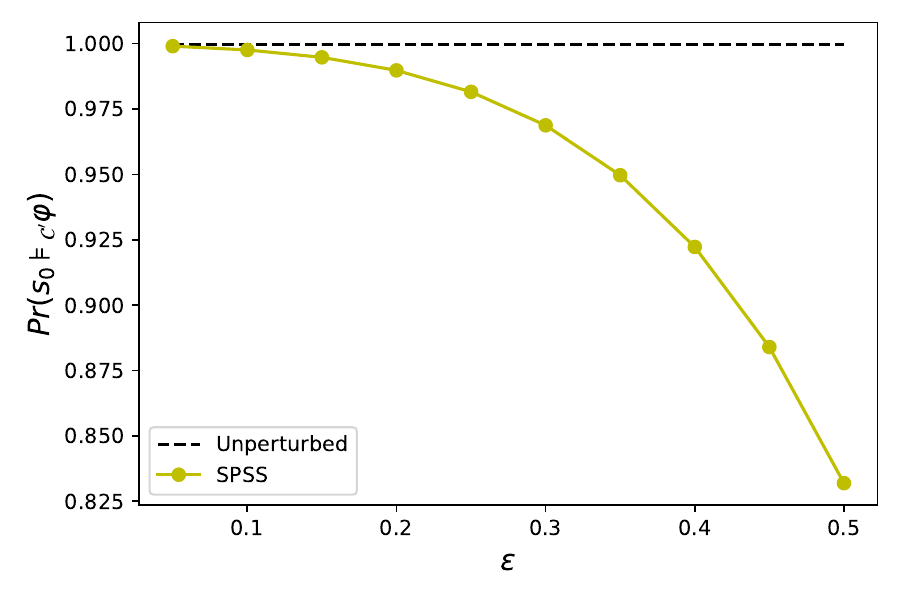}}
    \qquad
    \subfloat[Zeroconf IPv4 Protocol satisfaction probabilities under attack with $n=10$, $\varphi=\lozenge^{\leq 30}success$ and structure-preserving selected states (SPSS) threat models with vulnerable state sets $\{1,\dots5\}$ (\emph{early}), $\{6,\dots10\}$ (\emph{late}), and $\{1,\dots10\}$ (\emph{all}).\label{ipv4}]{%
    \includegraphics[width=.4\linewidth]{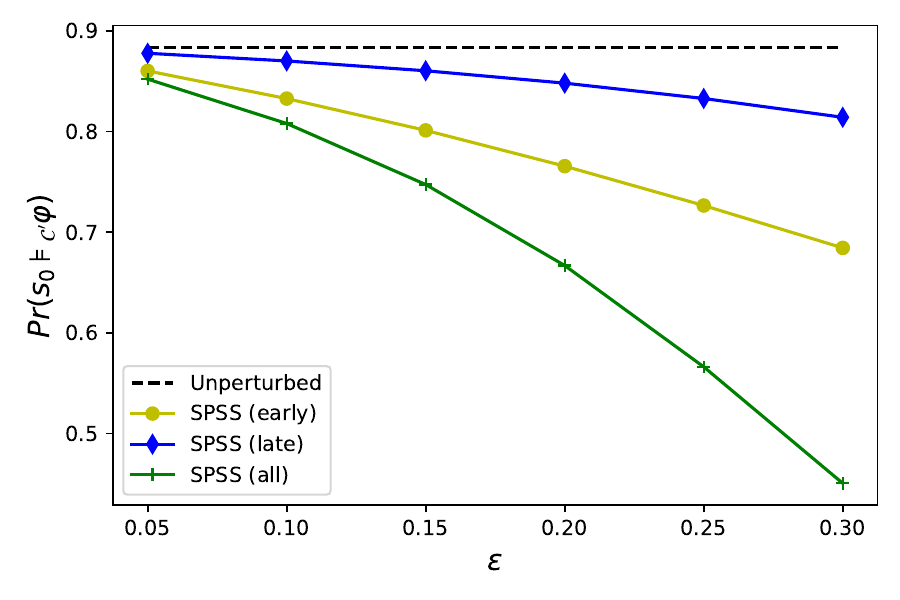}}
    \qquad
    \subfloat[$5\times 5$ Grid World satisfaction probabilities under attack with selected states (SS) and structure-preserving selected states (SPSS) threat models for vulnerable state sets $\{1, 3, 10\}$, and $\{1, 3, 10, 12, 19\}$ and $\varphi=\lozenge^{\leq 200}goal$. \label{gw1}]{%
    \includegraphics[width=.4\linewidth]{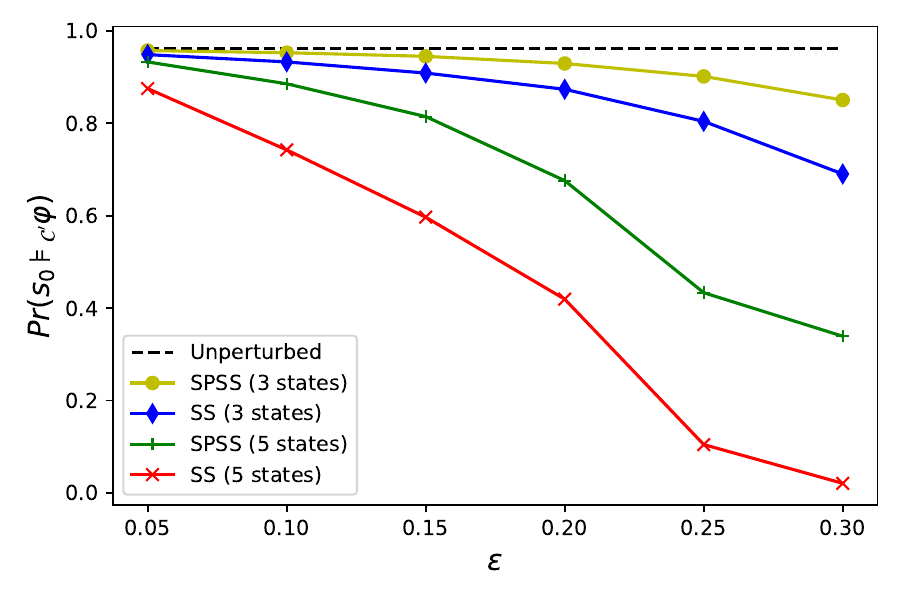}}
    \qquad
    \subfloat[$5\times 5$ Grid World satisfaction probabilities under attack with selected transitions (ST) and structure-preserving selected transitions (SPST) threat models attacking all adjacent transitions between the bottom right $3\times 3$ states, and bottom right $2\times 3$ states for $\varphi=\lozenge^{\leq 200}goal$.\label{gw2}]{%
    \includegraphics[width=.4\linewidth]{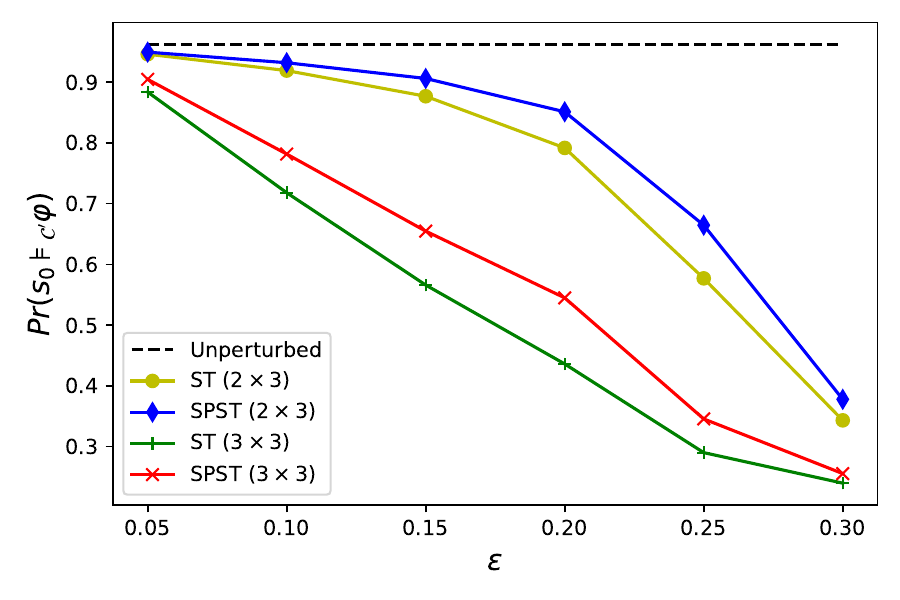}}
    \caption{Property satisfaction probability over attack budget (indicated by $\varepsilon$) over various case studies, threat models, and properties.}
    \label{fig:eps}
\end{figure*}

\subsection{Attack Performance}\label{sec:att_perf}
In Fig. \ref{fig:eps}, we show the property satisfaction probability for each of our three case studies against different properties and $\varepsilon,max$-bounded threat models of varying budget (indicated by $\varepsilon$). We choose values of $\varepsilon$ mostly between $0.05$ and $0.5$ to show the impact of attacker budget on attacker performance. We choose only small values of $\varepsilon$, to reflect the attackers' desire to make un-detectable modifications to the system. In the simple communication protocol, we only have one available attack state ($try$). We plot its performance over increasing $\varepsilon$ for $\varphi=\lozenge^{\leq 10}delivered$. As we increase $\varepsilon$, the property satisfaction probability decreases.

For the IPv4 Zeroconf Protocol we set $n=10$ and $m=50,000$ and measure satisfaction probability for $\varphi=\lozenge^{\leq 30}success$. We consider three instances of the structure-preserving selected states (SPSS) threat models vulnerable state sets $\{1,\dots5\}$, $\{6,\dots10\}$, and $\{1,\dots10\}$. We see that attacking the network at earlier time ticks is more effective than later, but attacking the network at all time ticks is most effective. As we increase $\varepsilon$, we see that all attacks decrease the property satisfaction probability, however the \emph{early} and \emph{all} attacks decrease this probability more quickly than the \emph{late} attack.

\begin{figure*}[t]
    \centering
    \includegraphics[width=.83\linewidth]{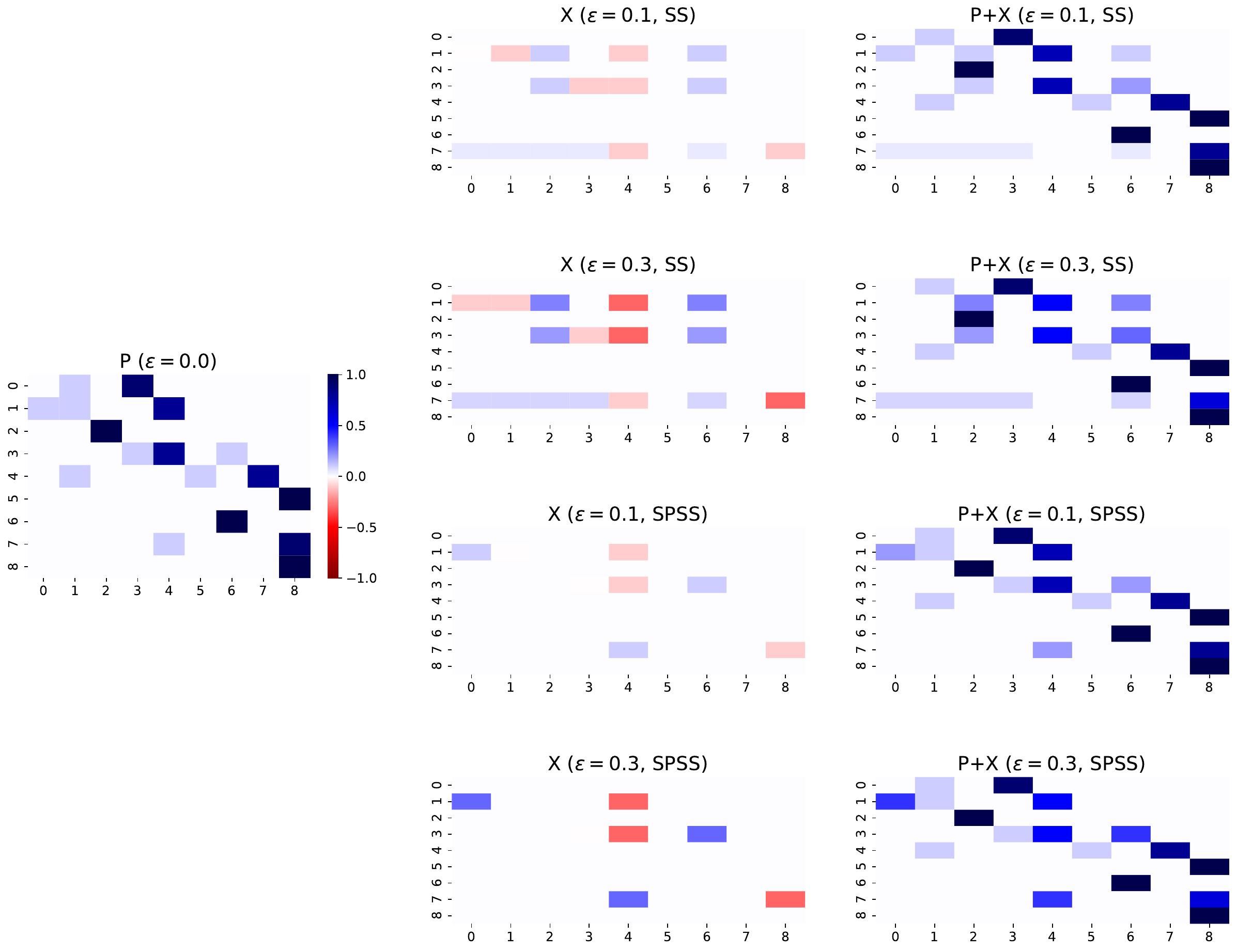}
    \caption{Transition probability matrix for the original DTMC (left), perturbation matrices for various threat models (middle) and perturbed DTMCs (right) with respect to $\varphi= (\neg hazard) \mathbf{U}^{\leq6} goal$. Specifically, selected states (SS) and structure-preserving selected states (SPSS) threat models with vulnerable states 1, 3, and 7 for $\varepsilon=0.1$ and $\varepsilon=0.3$. In general, perturbations increase probabilities of transitioning to hazard states and decrease probability of transitioning to the goal state. In the SS case, the previously unreachable hazard state becomes reachable after attack. Only rows 1, 3, and 7 of the perturbation matrices have non-zero entries.}
    \label{fig:synth}
\end{figure*}

For the randomized $5\times 5$ Grid World DTMCs we measure satisfaction probability for $\varphi=\lozenge^{\leq 200}goal$, where the $goal$ state is in the top right corner of the grid. Here, we consider all four $\varepsilon,max$-bounded threat models: selected states (SS), structure-preserving selected states (SPSS), selected transitions (ST), and structure-preserving selected transitions (SPST). We compare SS and SPSS attackers with vulnerable state sets $\{1, 3, 10\}$, and $\{1, 3, 10, 12, 19\}$, and ST and SPST attackers who can target all adjacent transitions in the bottom right $3\times 3$ states, and bottom right $2\times 3$ states. ST attackers can additionally add diagonally adjacent transitions in their respective regions. We see in Fig. \ref{gw1} and \ref{gw2}, as the attacker is able to control more states and transitions, the property satisfaction probability decreases more dramatically. In the case of $\varepsilon=.3$ and a non-structure-preserving attacker who controls $20\%$ of the states in the system, the property satisfaction probability drops from above $95\%$ to close to $0\%$. In contrast, a structure-preserving attacker who controls $12\%$ of states in the system can only decrease the property satisfaction probability to $~85\%$.

In all case studies, increasing $\varepsilon$ causes the property satisfaction probability to decrease, confirming Remark \ref{rem:mono}. Additionally, we see that, under the same class of case study, structure-preserving attacks are less effective than attacks which can add transitions to the system. Furthermore, adversaries acting under threat models with more vulnerable transitions perform better (i.e., cause more damage to system utility) than those with less.

\begin{figure*}[t]
    \centering
    \includegraphics[width=.6\linewidth]{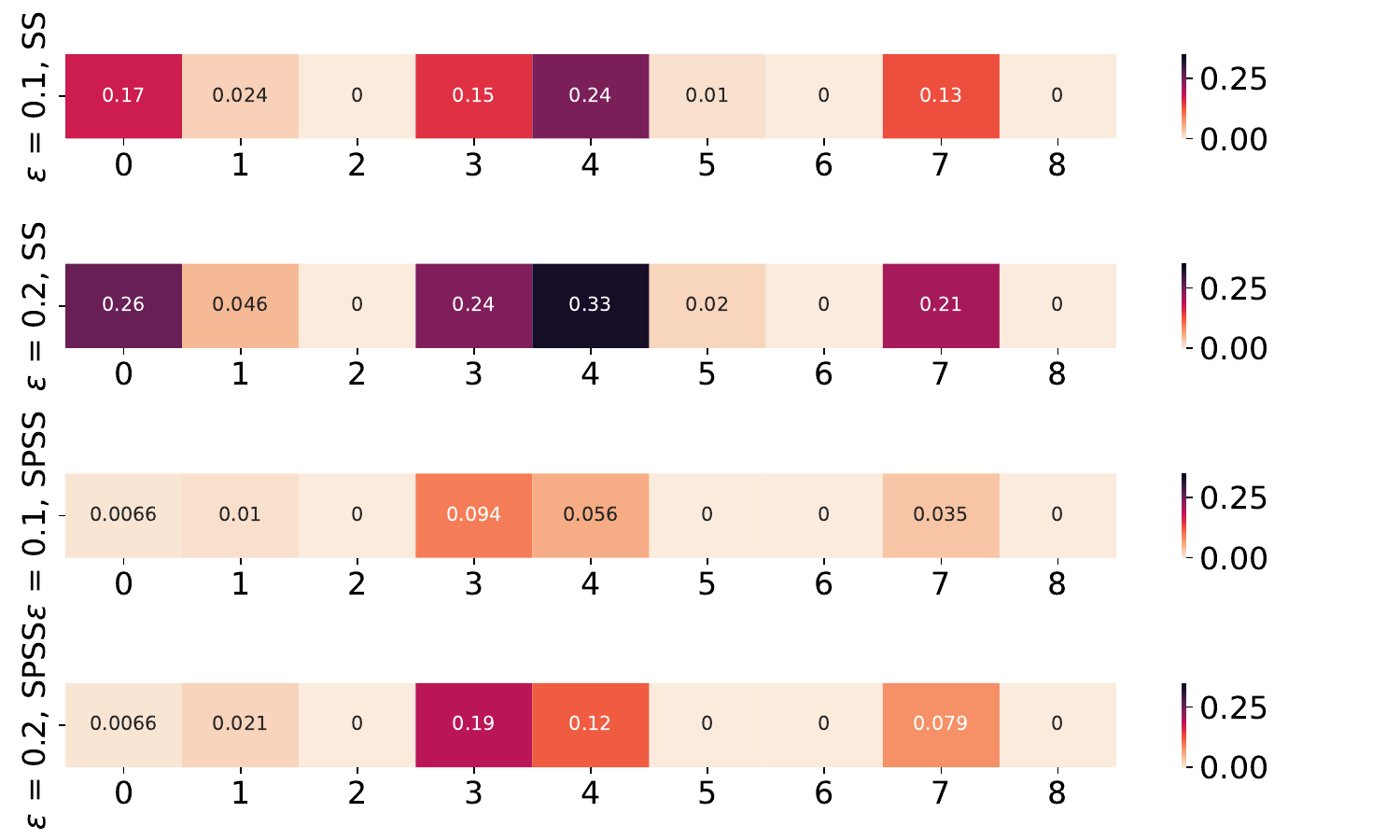}
    \caption{Horizontal axis indicates a unique state of the 9 state Grid World DTMC from Fig. \ref{gridworld}. Each row represents a different threat model and attack budget (indicated by $\varepsilon$). Darker colors indicate large $\delta$ between original and perturbed property satisfaction probability for the selected state (SS) or structure preserving selected state (SPSS) threat model with single vulnerable state for $\varphi = \neg hazard \mathbf{U}^{\leq6} goal$. Attacking certain states provides significantly more loss of utility than others. For example, state $3$ is near a hazard state, and attacking it causes a drop in utility across all threat models. State 1 is less valuable to the structure-preserving threat model than its counterpart. State 4 is central to the system and is valuable in both cases. Hazard and goal states are not valuable to the system with respect to $\varphi$, as being in these states means the property has already been satisfied or violated.}
    \label{fig:component}
\end{figure*}

\subsection{Attack Synthesis across Threat Models}
Using the $3\times3$ Grid World DTMC in Fig. \ref{gridworld}, we analyze the specific perturbations favored by attacks under the selected state (SS) threat model and  structure-preserving selected state (SPSS) threat model with vulnerable state set $\mathcal{V}=\{1,3,7\}$. In Fig. \ref{fig:synth} we present the worst-case adversarial perturbations under each threat model for $\varepsilon=0.1$ and $\varepsilon=0.3$ with respect to $\varphi= (\neg hazard) \mathbf{U}^{\leq6} goal$ as heatmaps of the perturbation matrix \pertmat. We choose these $\varepsilon$ values based on our observations of attack effectiveness vs. attacker budget in Sec. \ref{sec:att_perf}. Similarly to adversarial examples to neural networks, we see that small perturbations to the environment can lead to large changes in the system properties. Under all threat models, the attacker decreases the probability of jumping from state 7 to state 8, and redistributes those probabilities to states closer to the start. Attacks do not increase probability of reaching state 5, as it transitions to the goal state with probability 1. Under attacks which do not preserve structure, we see that the worst case attack adds transitions to hazard state 2, making it reachable from the start state, which was not the case in the original DTMC. All rows of the perturbation matrices which are not in the vulnerable state set have zero entries, as they cannot be modified. We infer that an attacker gets different success depending on the exact state they control in the system.

\subsection{Component Analysis}
Our robustness definitions are useful in determining which components of a system are important with respect to a specific property. In Fig. \ref{fig:component}, we compare $\delta$ between original and perturbed property satisfaction probability for SS and SPSS attackers with $\varepsilon=0.1$ and $\varepsilon=0.2$. Again, we choose these $\varepsilon$ values based on our observations of attack effectiveness vs. attacker budget in Sec. \ref{sec:att_perf}. For each threat model, we consider a set of one vulnerable attack state $\mathcal{V}=\{s\}$ for each state $s$ in a $3\times3$ Grid World DTMC from Fig. \ref{gridworld}. Under structure-preserving threat models, states with transitions to hazards and the goal are valuable. Under non structure-preserving threat models, the start state is also valuable.

\section{Connection to Adversarial Machine Learning}\label{sec:adv_ml}
There is an intuitive connection between our definition of adversarial robustness with respect to perturbed transition matrices, and the notion of adversarial robustness in supervised machine learning with respect to adversarial examples. While this is not a formal analysis, we believe this connection is useful in thinking about the problem, and was vital to our problem and solution formulation.

We demonstrate this connection by breaking down both DTMC modeling/verification and supervised machine learning (ML) model development into three main stages: training, validation, and deployment. We then describe how our perturbed transition probability matrices connect to adversarial examples for supervised machine learning. Finally, we explain how this breakdown leads to our optimization problem formulation and solutions for adversarial robustness with respect to probabilistic properties in a DTMC.

\subsection{Developing the Model}
In the \emph{training} stage of a supervised machine learning (ML) model for classification, a deterministic function (e.g., neural network, decision tree) is learned from labeled training data sampled from some distribution that is assumed to represent the real-world problem. The learning process often employs a probabilistic procedure such as Stochastic Gradient Descent (SGD). In the \emph{validation} stage, the performance of the trained model is evaluated by computing the loss between predicted and true labels of new data points sampled from the same distribution as the training data. In the \emph{deployment} stage, the model is used to perform classification tasks on real world data.

We explore a similar set of stages for modeling and verifying a system using a DTMC. In the training stage, the transition probability matrix of DTMC is specified or learned through observation of a system in the real world. In the validation stage, the system performance is calculated by computing the probability that the DTMC model satisfies some logical property. In the deployment stage, the system runs in the real world.

In both cases, we have a notion of training or learning a model, validation via evaluating some utility function, and running the system on the real world. This is not a perfect analogy, however. The key difference between the two cases is the validation stage. In the ML model, the validation stage runs the model on randomly sampled labeled data with the goal of evaluating how well the model classifies the underlying distribution. In contrast, for the DTMC, the validation stage computes the probability of satisfying a property in the DTMC with the goal of evaluating the quality of the system itself.

Though the goal of the validation is different, both scenarios assume the utility is the same in the deployment stage as it was in the validation stage. Therefore, both scenarios have a natural attacker who targets the deployment stage to cause unexpected behavior in the real-world environment. We find that thinking about the problem in these terms is useful for identifying a problem formulation and developing a solution.

\subsection{Modeling the Adversary}
We intuitively link our threat model formulation to adversarial examples to supervised ML models. An \emph{adversarial example} is an input in the deployment stage of an ML model, which is similar to a legitimate input from the assumed data distribution, but for which the model performs poorly \cite{biggio2013evasion,szegedy2013intriguing}. For example, one can generate an adversarial example by adding targeted noise bounded by an $\ell_\infty$ ball around a legitimate input to change the model's prediction.

We develop a notion similar to that of an adversarial example by defining valid perturbations to the transition probability matrix of the DTMC within an $\varepsilon$ ball with respect to a metric distance $d$. A system administrator can specify a minimum acceptable property satisfaction probability threshold (which we encode in the variable $\delta$) to categorize the system as functioning (satisfaction probability is above the threshold) or non-functioning (satisfaction probability is below the threshold). In both cases, the goal of the adversary is to maximally reduce system utility using small changes to the input. In an ML model, the attacker may want to change the prediction confidence on an input after applying a bounded perturbation. In the case of the DTMC, the attacker wants to maximally reduce the satisfaction probability of a property after applying a bounded perturbation to the transition probabilities.

The adversary capabilities for the two scenarios have a similarly analogous relationship. For ML models, an attacker who develops an adversarial example may provide deployment-time inputs to the model. In an effort to remain stealthy, the adversary has a bound on its ability to deviate from a legitimate input. We think of the ``input'' to the DTMC verification problem as the transition probability matrix. Of course, it is unlikely that an attacker can modify \emph{all} transitions in the matrix. However, it is possible for the attacker with access to some of these states and transitions to cause virtual or physical damage, as in the case of the collision avoidance protocol or warehouse robot from Sec. \ref{sec:intro}. Based on these observations, we can utilize techniques from adversarial machine learning to develop a solution to the adversarial robustness problem in DTMC modeling/validation.

\subsection{Finding a Worst-Case Attack using Optimization}
Our solutions are inspired by state-of-the art optimization techniques for finding adversarial examples in supervised machine learning models \cite{PSGD}. Similarly to the way that adversarial examples for ML models maximize the mis-classification confidence, our proposed formulation minimizes property satisfaction probability to render the system non-functioning. The main difference is the definition of the specific objective used in the optimization problem. Both problems utilize the bounds on the adversary as constraints in the optimization formulation of the problem. Importantly, both cases are deployment-time attacks and leverage the notion that there may be unexpected inputs when models are deployed in the real world.

\section{Related Work}
Recently, formal methods have been used for verifying adversarial robustness of machine learning models \cite{Reluplex,Reluval,AI2,verily,viper} and systems with machine learning components \cite{dreossi2019compositional,UT_Austin}. In the past few years, methods have been developed which utilize probabilistic model checking to verify properties of reinforcement learning models~\cite{gu_demonstration_2020,gotsman_deep_2020}. Our approach is motivated by research in adversarial machine learning on adversarial examples to neural networks \cite{biggio2013evasion,szegedy2013intriguing}.

Our problem definition is closely related to perturbation analysis of DTMCs \cite{abdulla_model_2011,baldan_perturbation_2014,hermanns_model-checking_2006,su_perturbation_2014,chechik_fact_2016}. There are many papers which explore different aspects of perturbation analysis in the non-adversarial case, focusing evaluation on simple PCTL or reachability properties and small parameter spaces. In \cite{abdulla_model_2011}, Bartocci \etal\ define the model repair problem where, given a property $\varphi$ and model $M$ which does not satisfy $\varphi$, the goal is to find a ``repaired'' model $M'$ which satisfies $\varphi$ and minimizes the cost of the perturbation from $M$ to $M'$. In \cite{baldan_perturbation_2014}, Chen \etal\ provide theoretical analysis of the maximal and minimal reachability probabilities for Markov chains within a bounded distance of one another. In~\cite{hermanns_model-checking_2006}, Sen \etal\ define a model checking problem over interval DTMCs for PCTL properties. In Section \ref{sec:umc} we provide a reduction to this problem for $\varepsilon,max$-bounded adversaries in the restricted PCTL setting. In \cite{su_perturbation_2014} Su and Rosenblum analyze systems using empirical distribution parameters, and in \cite{chechik_fact_2016} Calinescu \etal\ provide methods to compute confidence intervals for pDTMC properties. Perturbation analysis has also been studied in non-probabilistic systems~\cite{zhang_behavioral_2020,buccafurri_enhancing_1999}.

Our proposed solutions are based on explicit and symbolic DTMC model checking techniques \cite{prism,storm,ISCASMC,baier2008principles,courcoubetis_complexity_1995,lanotte2004decidability,daws_symbolic_2005,hahn2011synthesis}. Explicit probabilistic specification and model checking is well-researched, and there are many tools such as PRISM~\cite{prism}, Storm~\cite{storm}, ISCASMC~\cite{ISCASMC} and others which implement efficient algorithms for checking DTMCs on temporal logic properties~\cite{baier2008principles,courcoubetis_complexity_1995}. Symbolic property satisfaction function synthesis is a newer problem, originally addressed by Lanotte \etal, who provided theoretical assurances for pDTMCs with 1-2 parameters, and Daws, who provides a state-elimination algorithm for pDTMCs with non-nested PCTL properties~\cite{lanotte2004decidability,daws_symbolic_2005}. More recent methods have been developed to improve this computation, and these techniques have been implemented in tools such as PARAM, Prism, Storm, and PROPhESY~\cite{param,prism,prophesy,storm,hahn_probabilistic_2011,junges2019parameter,lanotte2007parametric}.

Another related area of study is in parameter synthesis for pDTMCs \cite{winkler_complexity_2019,quatmann_parameter_2016,prophesy,junges2019parameter,cubuktepe_synthesis_2018}. The parameter synthesis problem takes a pDTMC and temporal logic property, and attempts to partition the parameter space into accepting and non-accepting regions. Early work in this area primarily concerned reachability properties, and could only handle small parameter spaces~\cite{winkler_complexity_2019,quatmann_parameter_2016,prophesy}. In the past few years, new approaches have improved these methods to handle larger parameter spaces and more PCTL properties~\cite{junges2019parameter,cubuktepe_synthesis_2018}.

We also note the area of robust policy synthesis and verification in uncertain MDPs. In these problems, given a temporal logic property and parametric MDP, the goal is to synthesize a policy which can satisfy the property for all instantiations of the parametric MDP within some constraints. Several solutions have been developed which solve this problem for various logics and parameter spaces~\cite{nilim_robust_2005,puggelli_polynomial-time_nodate,wolff_robust_2012,delgado_efficient_2011,hahn2011synthesis}.
\section{Conclusion}
In this paper, we used existing probabilistic modeling and verification approaches to develop a formal framework of adversarial robustness in stochastic systems, including DTMCs and agents acting in MDPs.  Our framework allows us to formally define various threat models. Specifically, we provide four threat models under which an adversary can perturb the transition probabilities in an $\varepsilon$ ball around the
original system transitions. Two of these threat models preserve the structure of the underlying system, while the remaining two lift the structure-preserving assumption.

We use this notion of adversarial robustness to define three DTMC adversarial robustness problems: adversarial robustness verification, maximal $\delta$ synthesis, and worst case attack synthesis. We provide two optimization solutions. The first uses direct computation of the property satisfaction probability at each iteration of the optimization. The second pre-computes a symbolic representation of the property satisfaction probability and then uses that as the objective to the optimization. We found that utilizing symbolic model checking techniques can result in fast optimization time for small numbers of states and parameters, however as the parameter space increases in size, the symbolic solution becomes unmanageable. In the case of less restrictive adversarial models (and thus larger parameter spaces), direct computation of the property satisfaction probabilities is more scalable than the symbolic method. 

Our framework and solutions are useful to determine system robustness and to determine which components and transitions in a system are most critical to protect. We consider three case studies (two DTMCs and one composed DTMC from a policy and MDP), and compare attacks synthesized under each of the four threat models with respect to various PCTL$^*$ properties. We analyze in detail the resulting attack strategies, and provide component analysis of a system using our framework.

Future work can solve the problem of synthesizing robust systems with respect to these robustness definitions. Additionally, a potential line of future work is to use these definitions of adversarial robustness to verify systems with machine learning components including deep reinforcement learning.
\section*{Acknowledgment}
We would like to thank Professor David Parker and the anonymous reviewers for their helpful feedback on early drafts of this paper. Thanks also to Giorgio Severi, Matthew Jageilski, Ben Weintraub, Konstantina Bairaktari, and Max Von Hippel for their help on technical aspects of this work. Special thanks to Joyce Oakley, Shiyanbade Animashaun, and Di Wu, whose support through the COVID-19 pandemic was invaluable to the research and writing of this paper.

This work has been supported by the National Science Foundation under NSF SaTC awards CNS-1717634 and CNS-1801546.

\bibliographystyle{IEEEtran}
\bibliography{refs}

\begin{thebibliography}{10}
\providecommand{\url}[1]{#1}
\csname url@samestyle\endcsname
\providecommand{\newblock}{\relax}
\providecommand{\bibinfo}[2]{#2}
\providecommand{\BIBentrySTDinterwordspacing}{\spaceskip=0pt\relax}
\providecommand{\BIBentryALTinterwordstretchfactor}{4}
\providecommand{\BIBentryALTinterwordspacing}{\spaceskip=\fontdimen2\font plus
\BIBentryALTinterwordstretchfactor\fontdimen3\font minus
  \fontdimen4\font\relax}
\providecommand{\BIBforeignlanguage}[2]{{%
\expandafter\ifx\csname l@#1\endcsname\relax
\typeout{** WARNING: IEEEtran.bst: No hyphenation pattern has been}%
\typeout{** loaded for the language `#1'. Using the pattern for}%
\typeout{** the default language instead.}%
\else
\language=\csname l@#1\endcsname
\fi
#2}}
\providecommand{\BIBdecl}{\relax}
\BIBdecl

\bibitem{KNP12a}
M.~Kwiatkowska, G.~Norman, and D.~Parker, ``Probabilistic verification of
  herman’s self-stabilisation algorithm,'' \emph{Formal Aspects of
  Computing}, vol.~24, no.~4, pp. 661--670, 2012.

\bibitem{DKNP06}
M.~Duflot, M.~Kwiatkowska, G.~Norman, and D.~Parker, ``A formal analysis of
  {Bluetooth} device discovery,'' \emph{Int. Journal on Software Tools for
  Technology Transfer}, vol.~8, no.~6, pp. 621--632, 2006.

\bibitem{KN02}
M.~Kwiatkowska and G.~Norman, ``Verifying randomized {Byzantine} agreement,''
  in \emph{Proc. Formal Techniques for Networked and Distributed Systems
  (FORTE'02)}, ser. LNCS, D.~Peled and M.~Vardi, Eds., vol. 2529.\hskip 1em
  plus 0.5em minus 0.4em\relax Springer, 2002, pp. 194--209.

\bibitem{daws_symbolic_2005}
C.~Daws, ``\BIBforeignlanguage{en}{Symbolic and {Parametric} {Model} {Checking}
  of {Discrete}-{Time} {Markov} {Chains}},'' in
  \emph{\BIBforeignlanguage{en}{Theoretical {Aspects} of {Computing} - {ICTAC}
  2004}}.\hskip 1em plus 0.5em minus 0.4em\relax Berlin, Heidelberg: Springer
  Berlin Heidelberg, 2005, vol. 3407, pp. 280--294, series Title: Lecture Notes
  in Computer Science.

\bibitem{gu_demonstration_2020}
R.~Gu, ``\BIBforeignlanguage{en}{Demonstration of {TAMAA}},'' Jan. 2020,
  publisher: Zenodo.

\bibitem{gotsman_deep_2020}
T.~P. Gros, H.~Hermanns, J.~Hoffmann, M.~Klauck, and M.~Steinmetz,
  ``\BIBforeignlanguage{en}{Deep {Statistical} {Model} {Checking}},'' in
  \emph{\BIBforeignlanguage{en}{Formal {Techniques} for {Distributed}
  {Objects}, {Components}, and {Systems}}}, A.~Gotsman and A.~Sokolova,
  Eds.\hskip 1em plus 0.5em minus 0.4em\relax Cham: Springer International
  Publishing, 2020, vol. 12136, pp. 96--114, series Title: Lecture Notes in
  Computer Science.

\bibitem{storm}
C.~Dehnert, S.~Junges, J.-P. Katoen, and M.~Volk, ``A storm is coming: A modern
  probabilistic model checker,'' in \emph{Computer Aided Verification},
  R.~Majumdar and V.~Kun{\v{c}}ak, Eds.\hskip 1em plus 0.5em minus 0.4em\relax
  Cham: Springer International Publishing, 2017, pp. 592--600.

\bibitem{param}
E.~M. Hahn, H.~Hermanns, B.~Wachter, and L.~Zhang,
  ``\BIBforeignlanguage{en}{{PARAM}: {A} {Model} {Checker} for {Parametric}
  {Markov} {Models}},'' in \emph{\BIBforeignlanguage{en}{Computer {Aided}
  {Verification}}}.\hskip 1em plus 0.5em minus 0.4em\relax Berlin, Heidelberg:
  Springer Berlin Heidelberg, 2010, vol. 6174, pp. 660--664, series Title:
  Lecture Notes in Computer Science.

\bibitem{prophesy}
C.~Dehnert, S.~Junges, N.~Jansen, F.~Corzilius, M.~Volk, H.~Bruintjes, J.-P.
  Katoen, and E.~Ábrahám, ``\BIBforeignlanguage{en}{{PROPhESY}: {A}
  {PRObabilistic} {ParamEter} {SYnthesis} {Tool}},'' in
  \emph{\BIBforeignlanguage{en}{Computer {Aided} {Verification}}}.\hskip 1em
  plus 0.5em minus 0.4em\relax Cham: Springer International Publishing, 2015,
  vol. 9206, pp. 214--231, series Title: Lecture Notes in Computer Science.

\bibitem{winkler_complexity_2019}
T.~Winkler, S.~Junges, G.~A. Pérez, and J.-P. Katoen,
  ``\BIBforeignlanguage{en}{On the {Complexity} of {Reachability} in
  {Parametric} {Markov} {Decision} {Processes}},'' p. 17 pages, 2019, artwork
  Size: 17 pages Medium: application/pdf Publisher: Schloss Dagstuhl -
  Leibniz-Zentrum fuer Informatik GmbH, Wadern/Saarbruecken, Germany Version
  Number: 1.0.

\bibitem{quatmann_parameter_2016}
T.~Quatmann, C.~Dehnert, N.~Jansen, S.~Junges, and J.-P. Katoen,
  ``\BIBforeignlanguage{en}{Parameter {Synthesis} for {Markov} {Models}:
  {Faster} {Than} {Ever}},'' \emph{\BIBforeignlanguage{en}{arXiv:1602.05113
  [cs]}}, May 2016, arXiv: 1602.05113.

\bibitem{junges2019parameter}
S.~Junges, E.~{\'A}brah{\'a}m, C.~Hensel, N.~Jansen, J.-P. Katoen, T.~Quatmann,
  and M.~Volk, ``Parameter synthesis for markov models,'' \emph{arXiv preprint
  arXiv:1903.07993}, 2019.

\bibitem{cubuktepe_synthesis_2018}
M.~Cubuktepe, N.~Jansen, S.~Junges, J.-P. Katoen, and U.~Topcu,
  ``\BIBforeignlanguage{en}{Synthesis in {pMDPs}: {A} {Tale} of 1001
  {Parameters}},'' \emph{\BIBforeignlanguage{en}{arXiv:1803.02884 [cs, math]}},
  Jul. 2018, arXiv: 1803.02884.

\bibitem{KNSW07}
M.~Kwiatkowska, G.~Norman, J.~Sproston, and F.~Wang, ``Symbolic model checking
  for probabilistic timed automata,'' \emph{Information and Computation}, vol.
  205, no.~7, pp. 1027--1077, 2007.

\bibitem{KNS02a}
M.~Kwiatkowska, G.~Norman, and J.~Sproston, ``Probabilistic model checking of
  the {IEEE} 802.11 wireless local area network protocol,'' in \emph{Proc. 2nd
  Joint International Workshop on Process Algebra and Probabilistic Methods,
  Performance Modeling and Verification (PAPM/PROBMIV'02)}, ser. LNCS,
  H.~Hermanns and R.~Segala, Eds., vol. 2399.\hskip 1em plus 0.5em minus
  0.4em\relax Springer, 2002, pp. 169--187.

\bibitem{Fru11}
M.~Fruth, ``Formal methods for the analysis of wireless network protocols,''
  Ph.D. dissertation, Oxford University, 2011.

\bibitem{gridworld}
M.~L. Littman, A.~R. Cassandra, and L.~P. Kaelbling, ``Learning policies for
  partially observable environments: Scaling up,'' in \emph{Machine Learning
  Proceedings 1995}.\hskip 1em plus 0.5em minus 0.4em\relax Elsevier, 1995, pp.
  362--370.

\bibitem{hermanns_model-checking_2006}
K.~Sen, M.~Viswanathan, and G.~Agha, ``\BIBforeignlanguage{en}{Model-{Checking}
  {Markov} {Chains} in the {Presence} of {Uncertainties}},'' in
  \emph{\BIBforeignlanguage{en}{Tools and {Algorithms} for the {Construction}
  and {Analysis} of {Systems}}}.\hskip 1em plus 0.5em minus 0.4em\relax Berlin,
  Heidelberg: Springer Berlin Heidelberg, 2006, vol. 3920, pp. 394--410, series
  Title: Lecture Notes in Computer Science.

\bibitem{2020SciPy-NMeth}
P.~Virtanen, R.~Gommers, T.~E. Oliphant, M.~Haberland, T.~Reddy, D.~Cournapeau,
  E.~Burovski, P.~Peterson, W.~Weckesser, J.~Bright, S.~J. {van der Walt},
  M.~Brett, J.~Wilson, K.~J. Millman, N.~Mayorov, A.~R.~J. Nelson, E.~Jones,
  R.~Kern, E.~Larson, C.~J. Carey, {\.I}.~Polat, Y.~Feng, E.~W. Moore,
  J.~{VanderPlas}, D.~Laxalde, J.~Perktold, R.~Cimrman, I.~Henriksen, E.~A.
  Quintero, C.~R. Harris, A.~M. Archibald, A.~H. Ribeiro, F.~Pedregosa, P.~{van
  Mulbregt}, and {SciPy 1.0 Contributors}, ``{{SciPy} 1.0: Fundamental
  Algorithms for Scientific Computing in Python},'' \emph{Nature Methods},
  vol.~17, pp. 261--272, 2020.

\bibitem{prism}
M.~Kwiatkowska, G.~Norman, and D.~Parker, ``{PRISM} 4.0: Verification of
  probabilistic real-time systems,'' in \emph{Proc. 23rd International
  Conference on Computer Aided Verification (CAV'11)}, ser. LNCS,
  G.~Gopalakrishnan and S.~Qadeer, Eds., vol. 6806.\hskip 1em plus 0.5em minus
  0.4em\relax Springer, 2011, pp. 585--591.

\bibitem{baier2008principles}
C.~Baier and J.-P. Katoen, \emph{Principles of model checking}.\hskip 1em plus
  0.5em minus 0.4em\relax MIT press, 2008.

\bibitem{biggio2013evasion}
B.~Biggio, I.~Corona, D.~Maiorca, B.~Nelson, N.~{\v{S}}rndi{\'c}, P.~Laskov,
  G.~Giacinto, and F.~Roli, ``Evasion attacks against machine learning at test
  time,'' in \emph{Joint European conference on machine learning and knowledge
  discovery in databases}.\hskip 1em plus 0.5em minus 0.4em\relax Springer,
  2013, pp. 387--402.

\bibitem{szegedy2013intriguing}
C.~Szegedy, W.~Zaremba, I.~Sutskever, J.~Bruna, D.~Erhan, I.~Goodfellow, and
  R.~Fergus, ``Intriguing properties of neural networks,'' \emph{arXiv preprint
  arXiv:1312.6199}, 2013.

\bibitem{ISCASMC}
E.~M. Hahn, Y.~Li, S.~Schewe, A.~Turrini, and L.~Zhang, ``iscas m c: a
  web-based probabilistic model checker,'' in \emph{International Symposium on
  Formal Methods}.\hskip 1em plus 0.5em minus 0.4em\relax Springer, 2014, pp.
  312--317.

\bibitem{abdulla_model_2011}
E.~Bartocci, R.~Grosu, P.~Katsaros, C.~R. Ramakrishnan, and S.~A. Smolka,
  ``\BIBforeignlanguage{en}{Model {Repair} for {Probabilistic} {Systems}},'' in
  \emph{\BIBforeignlanguage{en}{Tools and {Algorithms} for the {Construction}
  and {Analysis} of {Systems}}}.\hskip 1em plus 0.5em minus 0.4em\relax Berlin,
  Heidelberg: Springer Berlin Heidelberg, 2011, vol. 6605, pp. 326--340, series
  Title: Lecture Notes in Computer Science.

\bibitem{junges_parameter_2019}
S.~Junges, E.~Abraham, C.~Hensel, N.~Jansen, J.-P. Katoen, T.~Quatmann, and
  M.~Volk, ``\BIBforeignlanguage{en}{Parameter {Synthesis} for {Markov}
  {Models}},'' \emph{\BIBforeignlanguage{en}{arXiv:1903.07993 [cs]}}, Mar.
  2019, arXiv: 1903.07993.

\bibitem{topcu2012}
U.~Topcu, N.~Ozay, J.~Liu, and R.~M. Murray, ``On synthesizing robust discrete
  controllers under modeling uncertainty,'' in \emph{Proceedings of the 15th
  ACM International Conference on Hybrid Systems: Computation and Control},
  ser. HSCC '12.\hskip 1em plus 0.5em minus 0.4em\relax New York, NY, USA:
  Association for Computing Machinery, 2012, p. 85–94.

\bibitem{py4j}
``Py4j - a bridge between python and java.'' https://www.py4j.org/, 2021.

\bibitem{PSGD}
N.~Carlini and D.~Wagner, ``Towards evaluating the robustness of neural
  networks,'' in \emph{2017 ieee symposium on security and privacy (sp)}.\hskip
  1em plus 0.5em minus 0.4em\relax IEEE, 2017, pp. 39--57.

\bibitem{Reluplex}
G.~Katz, C.~Barrett, D.~L. Dill, K.~Julian, and M.~J. Kochenderfer, ``Reluplex:
  An efficient smt solver for verifying deep neural networks,'' in
  \emph{International Conference on Computer Aided Verification}.\hskip 1em
  plus 0.5em minus 0.4em\relax Springer, 2017, pp. 97--117.

\bibitem{Reluval}
S.~Wang, K.~Pei, J.~Whitehouse, J.~Yang, and S.~Jana, ``Formal security
  analysis of neural networks using symbolic intervals,'' in \emph{27th USENIX
  Security Symposium}, 2018, pp. 1599--1614.

\bibitem{AI2}
T.~{Gehr}, M.~{Mirman}, D.~{Drachsler-Cohen}, P.~{Tsankov}, S.~{Chaudhuri}, and
  M.~{Vechev}, ``Ai2: Safety and robustness certification of neural networks
  with abstract interpretation,'' in \emph{2018 IEEE Symposium on Security and
  Privacy (SP)}, 2018, pp. 3--18.

\bibitem{verily}
Y.~Kazak, C.~Barrett, G.~Katz, and M.~Schapira, ``Verifying deep-{RL}-driven
  systems,'' in \emph{Proceedings of the 2019 Workshop on Network Meets AI \&
  ML}, ser. NetAI’19.\hskip 1em plus 0.5em minus 0.4em\relax New York, NY,
  USA: Association for Computing Machinery, 2019, p. 83–89.

\bibitem{viper}
O.~Bastani, Y.~Pu, and A.~Solar-Lezama, ``Verifiable reinforcement learning via
  policy extraction,'' in \emph{Advances in Neural Information Processing
  Systems 31}, S.~Bengio, H.~Wallach, H.~Larochelle, K.~Grauman,
  N.~Cesa-Bianchi, and R.~Garnett, Eds.\hskip 1em plus 0.5em minus 0.4em\relax
  Curran Associates, Inc., 2018, pp. 2494--2504.

\bibitem{dreossi2019compositional}
T.~Dreossi, A.~Donz{\'e}, and S.~A. Seshia, ``Compositional falsification of
  cyber-physical systems with machine learning components,'' \emph{Journal of
  Automated Reasoning}, vol.~63, no.~4, pp. 1031--1053, 2019.

\bibitem{UT_Austin}
G.~Anderson, S.~Pailoor, I.~Dillig, and S.~Chaudhuri, ``Optimization and
  abstraction: A synergistic approach for analyzing neural network
  robustness,'' in \emph{Proceedings of the 40th ACM SIGPLAN Conference on
  Programming Language Design and Implementation}, ser. PLDI 2019.\hskip 1em
  plus 0.5em minus 0.4em\relax New York, NY, USA: Association for Computing
  Machinery, 2019, p. 731–744.

\bibitem{baldan_perturbation_2014}
T.~Chen, Y.~Feng, D.~S. Rosenblum, and G.~Su,
  ``\BIBforeignlanguage{en}{Perturbation {Analysis} in {Verification} of
  {Discrete}-{Time} {Markov} {Chains}},'' in
  \emph{\BIBforeignlanguage{en}{{CONCUR} 2014 – {Concurrency}
  {Theory}}}.\hskip 1em plus 0.5em minus 0.4em\relax Berlin, Heidelberg:
  Springer Berlin Heidelberg, 2014, vol. 8704, pp. 218--233, series Title:
  Lecture Notes in Computer Science.

\bibitem{su_perturbation_2014}
G.~Su and D.~S. Rosenblum, ``\BIBforeignlanguage{en}{Perturbation analysis of
  stochastic systems with empirical distribution parameters},'' in
  \emph{\BIBforeignlanguage{en}{Proceedings of the 36th {International}
  {Conference} on {Software} {Engineering}}}.\hskip 1em plus 0.5em minus
  0.4em\relax Hyderabad India: ACM, May 2014, pp. 311--321.

\bibitem{chechik_fact_2016}
R.~Calinescu, K.~Johnson, and C.~Paterson, ``\BIBforeignlanguage{en}{{FACT}:
  {A} {Probabilistic} {Model} {Checker} for {Formal} {Verification} with
  {Confidence} {Intervals}},'' in \emph{\BIBforeignlanguage{en}{Tools and
  {Algorithms} for the {Construction} and {Analysis} of {Systems}}}.\hskip 1em
  plus 0.5em minus 0.4em\relax Berlin, Heidelberg: Springer Berlin Heidelberg,
  2016, vol. 9636, pp. 540--546, series Title: Lecture Notes in Computer
  Science.

\bibitem{zhang_behavioral_2020}
C.~Zhang, D.~Garlan, and E.~Kang, ``\BIBforeignlanguage{en}{A behavioral notion
  of robustness for software systems},'' in
  \emph{\BIBforeignlanguage{en}{Proceedings of the 28th {ACM} {Joint} {Meeting}
  on {European} {Software} {Engineering} {Conference} and {Symposium} on the
  {Foundations} of {Software} {Engineering}}}.\hskip 1em plus 0.5em minus
  0.4em\relax Virtual Event USA: ACM, Nov. 2020, pp. 1--12.

\bibitem{buccafurri_enhancing_1999}
F.~Buccafurri, T.~Eiter, G.~Gottlob, and N.~Leone,
  ``\BIBforeignlanguage{en}{Enhancing model checking in verification by {AI}
  techniques},'' \emph{\BIBforeignlanguage{en}{Artificial Intelligence}}, vol.
  112, no. 1-2, pp. 57--104, Aug. 1999.

\bibitem{courcoubetis_complexity_1995}
C.~Courcoubetis and M.~Yannakakis, ``\BIBforeignlanguage{en}{The complexity of
  probabilistic verification},'' \emph{\BIBforeignlanguage{en}{Journal of the
  ACM}}, vol.~42, no.~4, pp. 857--907, Jul. 1995.

\bibitem{lanotte2004decidability}
R.~Lanotte, A.~Maggiolo-Schettini, and A.~Troina, ``Decidability results for
  parametric probabilistic transition systems with an application to
  security,'' in \emph{Proceedings of the Second International Conference on
  Software Engineering and Formal Methods, 2004. SEFM 2004.}\hskip 1em plus
  0.5em minus 0.4em\relax IEEE, 2004, pp. 114--121.

\bibitem{hahn2011synthesis}
E.~M. Hahn, T.~Han, and L.~Zhang, ``Synthesis for pctl in parametric markov
  decision processes,'' in \emph{Nasa formal methods symposium}.\hskip 1em plus
  0.5em minus 0.4em\relax Springer, 2011, pp. 146--161.

\bibitem{hahn_probabilistic_2011}
E.~M. Hahn, H.~Hermanns, and L.~Zhang, ``\BIBforeignlanguage{en}{Probabilistic
  reachability for parametric {Markov} models},''
  \emph{\BIBforeignlanguage{en}{International Journal on Software Tools for
  Technology Transfer}}, vol.~13, no.~1, pp. 3--19, Jan. 2011.

\bibitem{lanotte2007parametric}
R.~Lanotte, A.~Maggiolo-Schettini, and A.~Troina, ``Parametric probabilistic
  transition systems for system design and analysis,'' \emph{Formal Aspects of
  Computing}, vol.~19, no.~1, pp. 93--109, 2007.

\bibitem{nilim_robust_2005}
A.~Nilim and L.~El~Ghaoui, ``\BIBforeignlanguage{en}{Robust {Control} of
  {Markov} {Decision} {Processes} with {Uncertain} {Transition} {Matrices}},''
  \emph{\BIBforeignlanguage{en}{Operations Research}}, vol.~53, no.~5, pp.
  780--798, Oct. 2005.

\bibitem{puggelli_polynomial-time_nodate}
A.~Puggelli, W.~Li, A.~L. Sangiovanni-Vincentelli, and S.~A. Seshia,
  ``\BIBforeignlanguage{en}{Polynomial-{Time} {Veriﬁcation} of {PCTL}
  {Properties} of {MDPs} with {Convex} {Uncertainties}},'' p.~16.

\bibitem{wolff_robust_2012}
E.~M. Wolff, U.~Topcu, and R.~M. Murray, ``\BIBforeignlanguage{en}{Robust
  control of uncertain {Markov} {Decision} {Processes} with temporal logic
  specifications},'' in \emph{\BIBforeignlanguage{en}{2012 {IEEE} 51st {IEEE}
  {Conference} on {Decision} and {Control} ({CDC})}}.\hskip 1em plus 0.5em
  minus 0.4em\relax Maui, HI, USA: IEEE, Dec. 2012, pp. 3372--3379.

\bibitem{delgado_efficient_2011}
K.~V. Delgado, S.~Sanner, and L.~N. de~Barros,
  ``\BIBforeignlanguage{en}{Efficient solutions to factored {MDPs} with
  imprecise transition probabilities},''
  \emph{\BIBforeignlanguage{en}{Artificial Intelligence}}, vol. 175, no. 9-10,
  pp. 1498--1527, Jun. 2011.

\end{thebibliography}

\end{document}